\documentclass[letterpaper]{article} 
\usepackage{aaai2026}  
\usepackage{times}  
\usepackage{helvet}  
\usepackage{courier}  
\usepackage[hyphens]{url}  
\usepackage{graphicx} 
\usepackage{natbib}  
\usepackage{caption} 
\usepackage{amsmath,amssymb,amsthm}
\usepackage{booktabs}
\usepackage{multirow}
\usepackage{diagbox}
\usepackage{makecell}
\usepackage{tabularx}
\usepackage{adjustbox}
\usepackage[table]{xcolor}
\usepackage{bm}
\usepackage{pifont}
\usepackage{algorithm}
\usepackage{algorithmic}
\usepackage{subcaption}
\usepackage{tikz}
\usetikzlibrary{fit,positioning,shapes,backgrounds}
\usepackage{enumitem}
\usepackage{lipsum}
\usepackage{cellspace}
\definecolor{mydarkblue}{rgb}{0,0.08,0.45}
\usepackage[utf8]{inputenc}
\usepackage[framemethod=tikz]{mdframed}
\usepackage[most]{tcolorbox}
\usepackage{xspace}
\usepackage{bbding}
\usepackage{fancyhdr}
\usepackage{newfloat}
\usepackage{listings}

\usepackage{enumitem}
\usepackage{xspace}
\usepackage{bbding}
\usepackage{fancyhdr}

\usetikzlibrary{fit}
\usetikzlibrary{positioning}
\usetikzlibrary{shapes,backgrounds}

\newtheorem{theorem}{Theorem}

\theoremstyle{remark}

\newcommand{\greencheck}{\textcolor{green}{\ding{51}}}  
\newcommand{\redcross}{\textcolor{red}{\ding{55}}}      
\DeclareCaptionStyle{ruled}{labelfont=normalfont,labelsep=colon,strut=off} 
\floatstyle{ruled}
\newfloat{listing}{tb}{lst}{}
\floatname{listing}{Listing}

\title{\emph{HarmoQ}: Harmonized Post-Training Quantization for High-Fidelity Image Super-Resolution}

\author{
	Hongjun Wang\textsuperscript{\rm 1,2},
	Jiyuan Chen\textsuperscript{\rm 2,3},
	Xuan Song\textsuperscript{\rm 4}\footnotemark[1],
	Yinqiang Zheng\textsuperscript{\rm 1}\thanks{Corresponding Author}
}
\affiliations{
	\textsuperscript{\rm 1}The University of Tokyo\\
	\textsuperscript{\rm 2}Southern University of Science and Technology\\
	\textsuperscript{\rm 3}The Hong Kong Polytechnic University\\
	\textsuperscript{\rm 4}Jilin University\\
}

\begin{document}
	
	\maketitle
\begin{abstract}
	Post-training quantization offers an efficient pathway to deploy super-resolution models, yet existing methods treat weight and activation quantization independently, missing their critical interplay. Through controlled experiments on SwinIR, we uncover a striking asymmetry: weight quantization primarily degrades structural similarity, while activation quantization disproportionately affects pixel-level accuracy. This stems from their distinct roles—weights encode learned restoration priors for textures and edges, whereas activations carry input-specific intensity information.
	Building on this insight, we propose HarmoQ, a unified framework that harmonizes quantization across components through three synergistic steps: structural residual calibration proactively adjusts weights to compensate for activation-induced detail loss, harmonized scale optimization analytically balances quantization difficulty via closed-form solutions, and adaptive boundary refinement iteratively maintains this balance during optimization. Experiments show HarmoQ achieves substantial gains under aggressive compression, outperforming prior art by 0.46 dB on Set5 at 2-bit while delivering 3.2× speedup and 4× memory reduction on A100 GPUs. This work provides the first systematic analysis of weight-activation coupling in super-resolution quantization and establishes a principled solution for efficient high-quality image restoration.
\end{abstract}

\begin{links}
	\link{Code}{https://github.com/Dreamzz5/HarmoQ}
\end{links}

	\section{Introduction}
	
	Single-image super-resolution (SISR) has witnessed remarkable progress driven by deep neural networks, especially Transformer-based architectures such as SwinIR~\cite{liang2021swinir} and Restormer~\cite{zamir2022restormer}. These models achieve impressive reconstruction accuracy but come at a steep cost: their high parameter counts and computational demands pose major barriers to deployment in real-world scenarios, particularly in edge or mobile environments.
	
	Post-training quantization (PTQ) offers an appealing solution to compress pretrained models without retraining. However, applying PTQ to super-resolution remains an open challenge. Prior works diverge on how to balance quantization: some target both weights and activations~\cite{liu20242dquant}, while others focus exclusively on activations~\cite{wang2025thinking}. Yet, the relative impact of these two components on visual fidelity remains poorly understood.
	
	\begin{figure}[t]
		\centering
		\subfloat[PSNR and SSIM  under weight and activation quantization]{\includegraphics[width=1\linewidth]{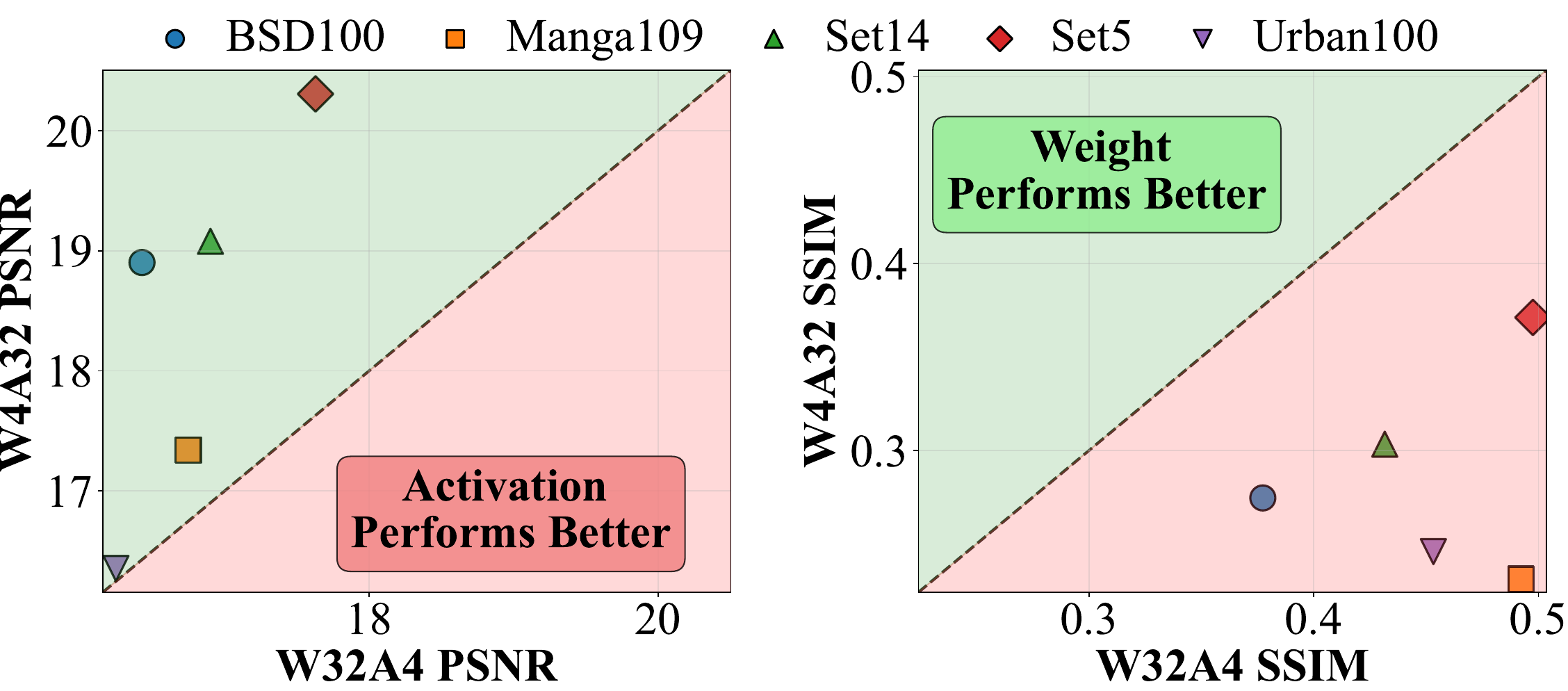}}
		
		\subfloat[Ground Truth]{\includegraphics[width=0.32\linewidth]{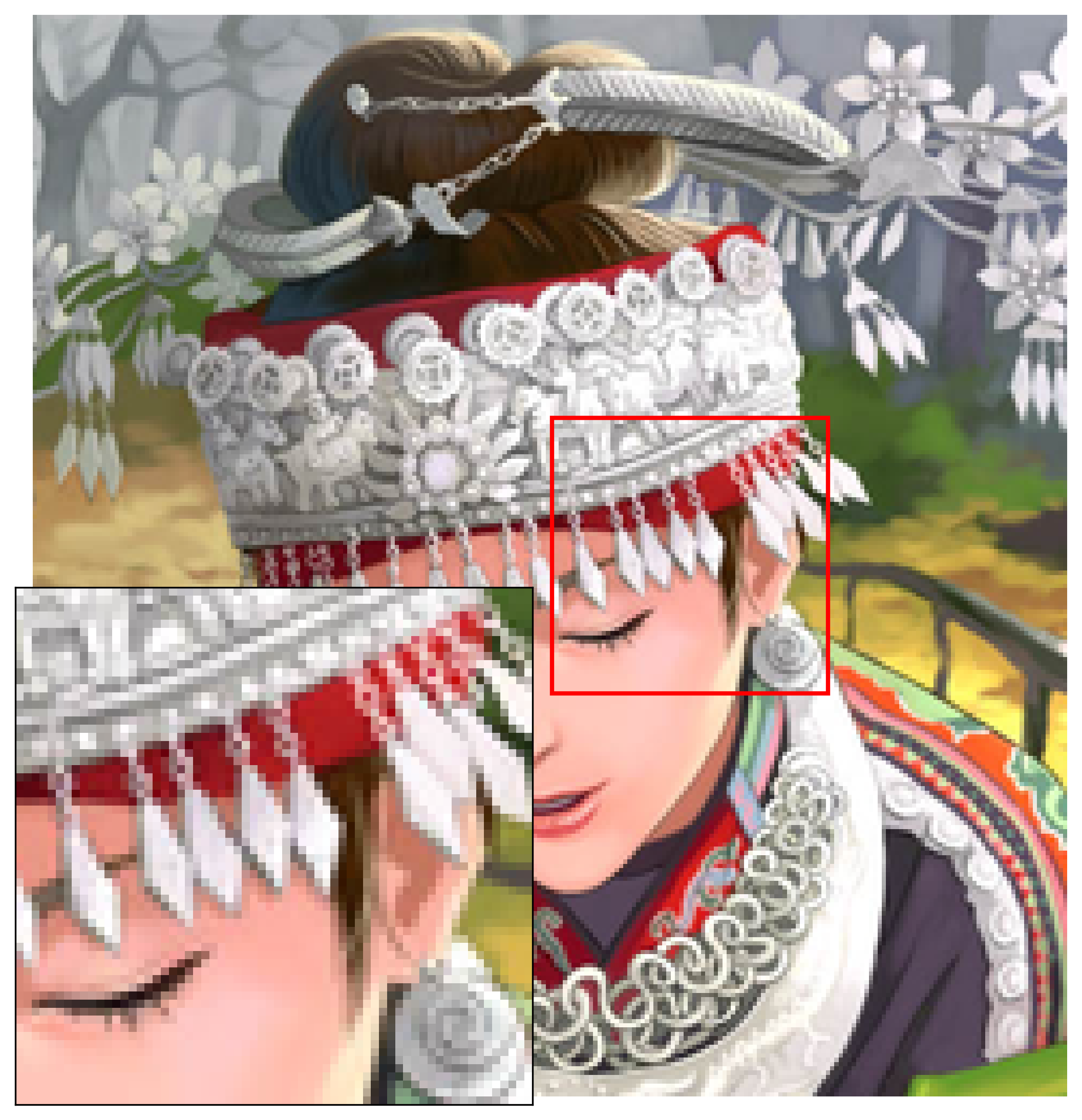}}
		\subfloat[W32A4 result]{\includegraphics[width=0.32\linewidth]{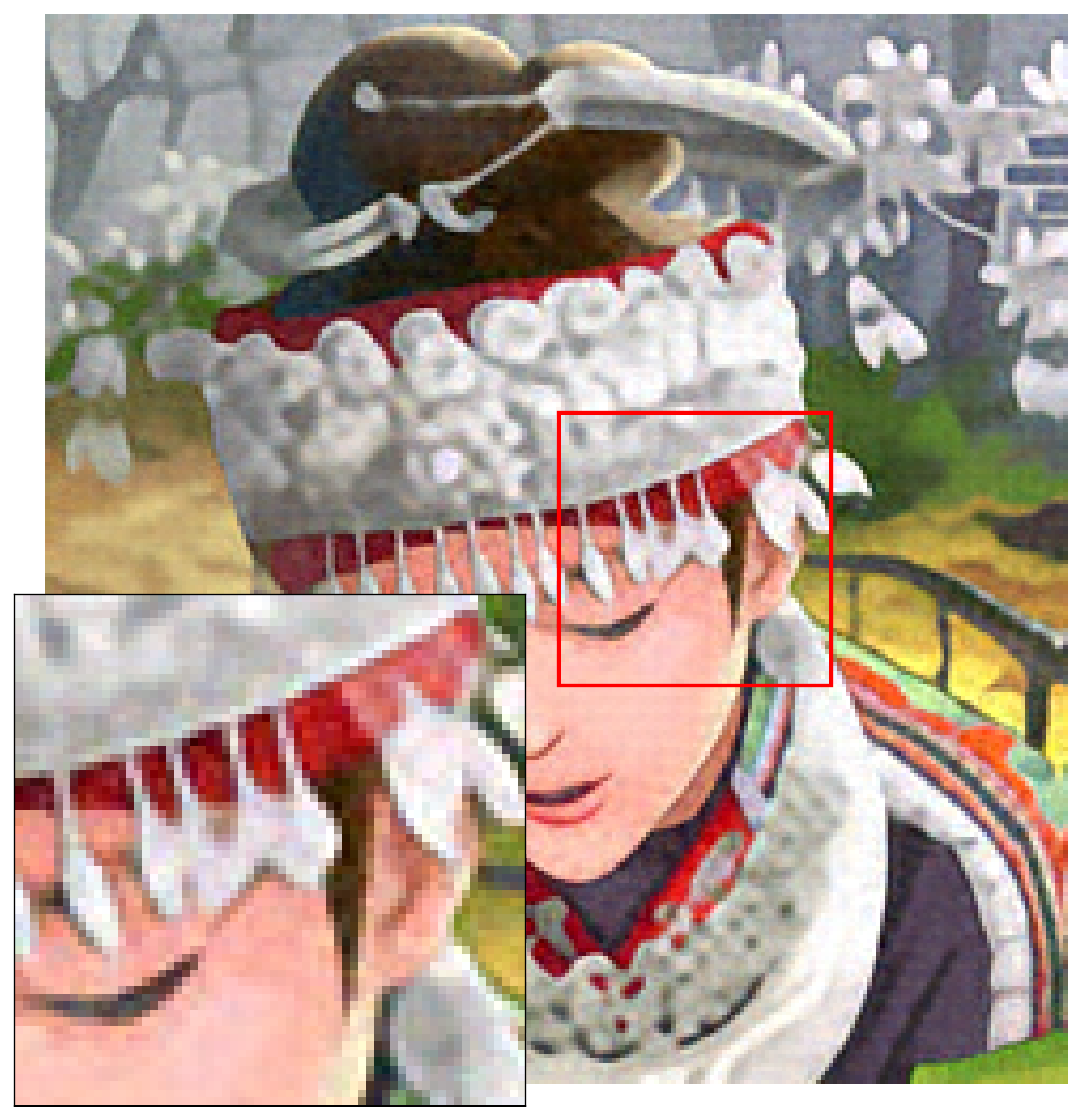}}
		\subfloat[W4 A32 result]{\includegraphics[width=0.32\linewidth]{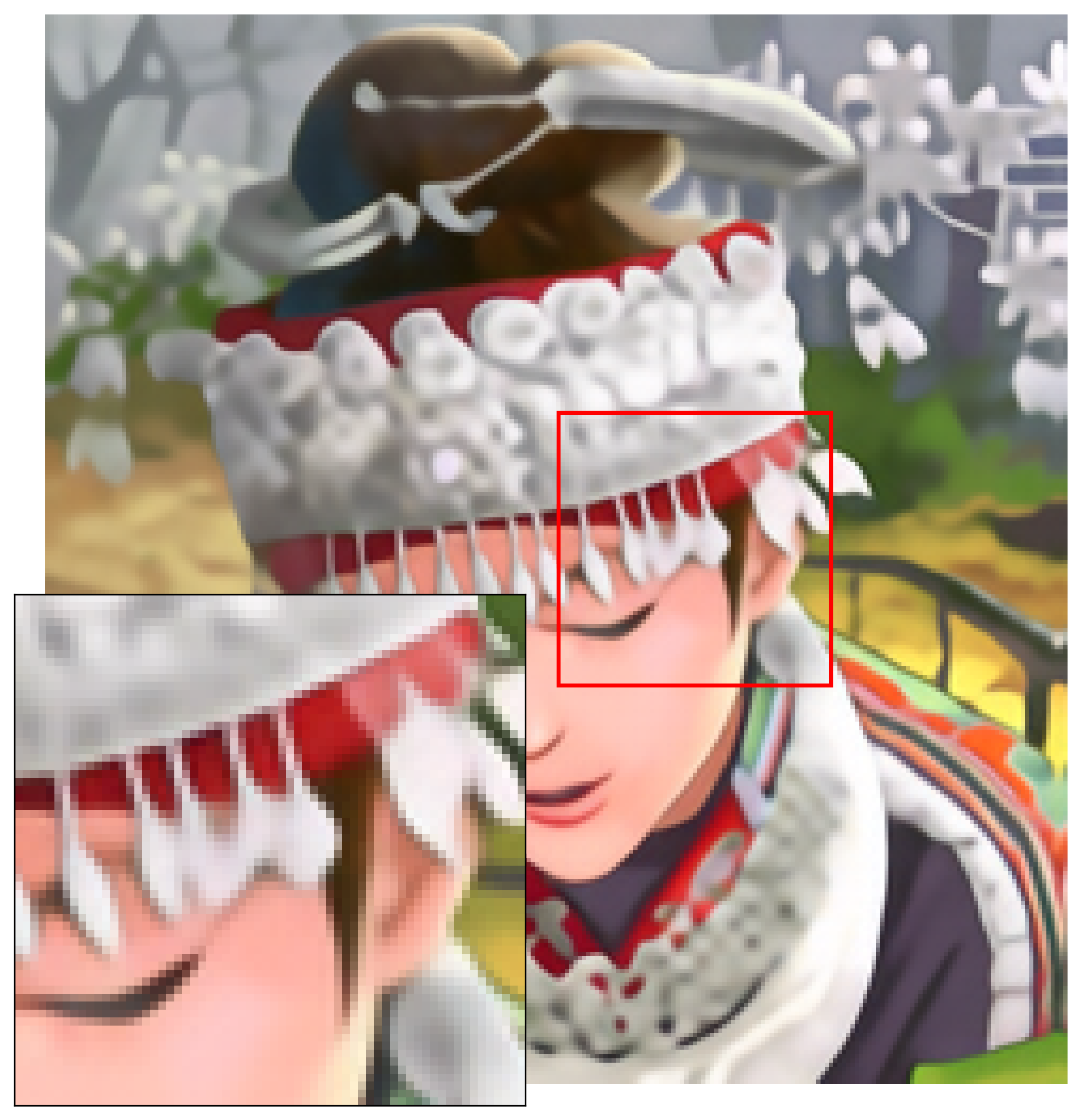}}
		\caption{\textbf{Weight vs. activation quantization analysis.}  (a) Performance comparison showing that W32A4 (activation quantization) primarily degrades PSNR while W4A32 (weight quantization) shows stronger SSIM degradation, indicating complementary effects on pixel-level accuracy versus structural similarity. (b, c) Visual comparison of W32A4 and W4A32 quantization results on anime image reconstruction.}
		\label{fig:weightactivatecompare}
		\vspace{-10pt}
	\end{figure}
	
	To systematically investigate this trade-off, we conduct a controlled comparison between two complementary 4-bit quantization regimes—W4A32 (quantized weights, full-precision activations) and W32A4 (full-precision weights, quantized activations)—on SwinIR across standard benchmarks. As shown in Figure~\ref{fig:weightactivatecompare}, a clear pattern emerges: 
	\textbf{weight quantization (W4A32)} has a stronger effect on SSIM (drops to 0.85 
	on Manga109), indicating corruption of fine-grained textures and edge structures, 
	while \textbf{activation quantization (W32A4)} disproportionately affects PSNR 
	(drops to 32.5 dB on Urban100), reflecting degradation of pixel-level accuracy and global attributes such as brightness and contrast.
	\begin{tcolorbox}[colback=blue!5, colframe=blue!50!black, boxrule=1pt, rounded corners]
		\textbf{Key Observation:} Weight quantization primarily impacts structural similarity and texture fidelity (reflected in SSIM degradation). Activation quantization mainly affects pixel-level accuracy and global appearance attributes (reflected in PSNR loss). 
	\end{tcolorbox}
	This asymmetry arises from the distinct functional roles of weights and 
	activations. Weights encode distributed, global restoration priors learned 
	during training; their quantization introduces systematic errors that primarily 
	corrupt representations of fine textures and edges. In contrast, 
	activations carry local, input-specific intensity information; their quantization 
	introduces stochastic noise that disproportionately affects pixel-level accuracy, degrading global attributes like brightness uniformity and overall contrast.
	
	\begin{figure}[t]
		\centering
		\includegraphics[width=1\linewidth]{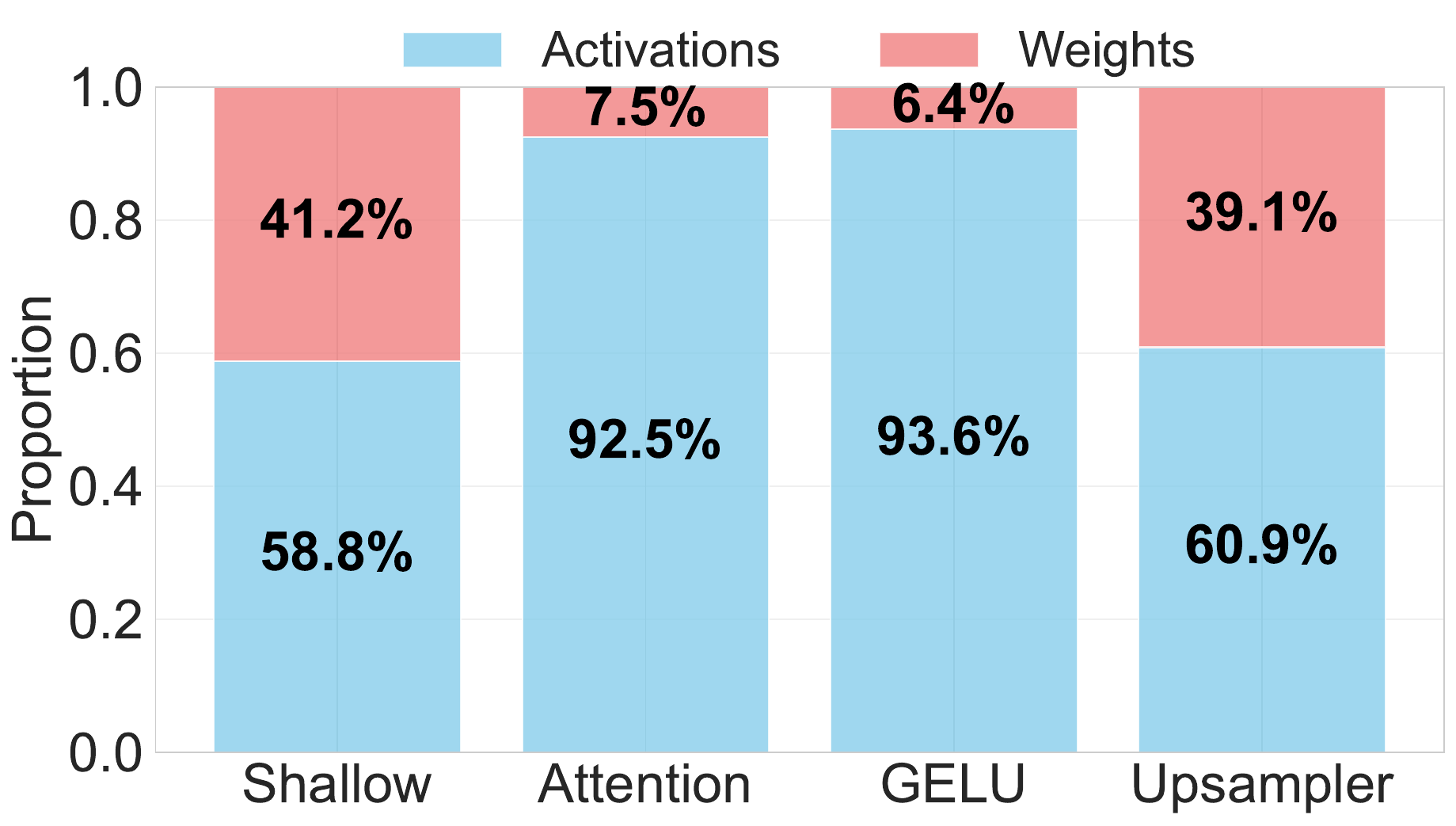}
		\subfloat[W32A4 Attention Map]{\includegraphics[width=0.48\linewidth]{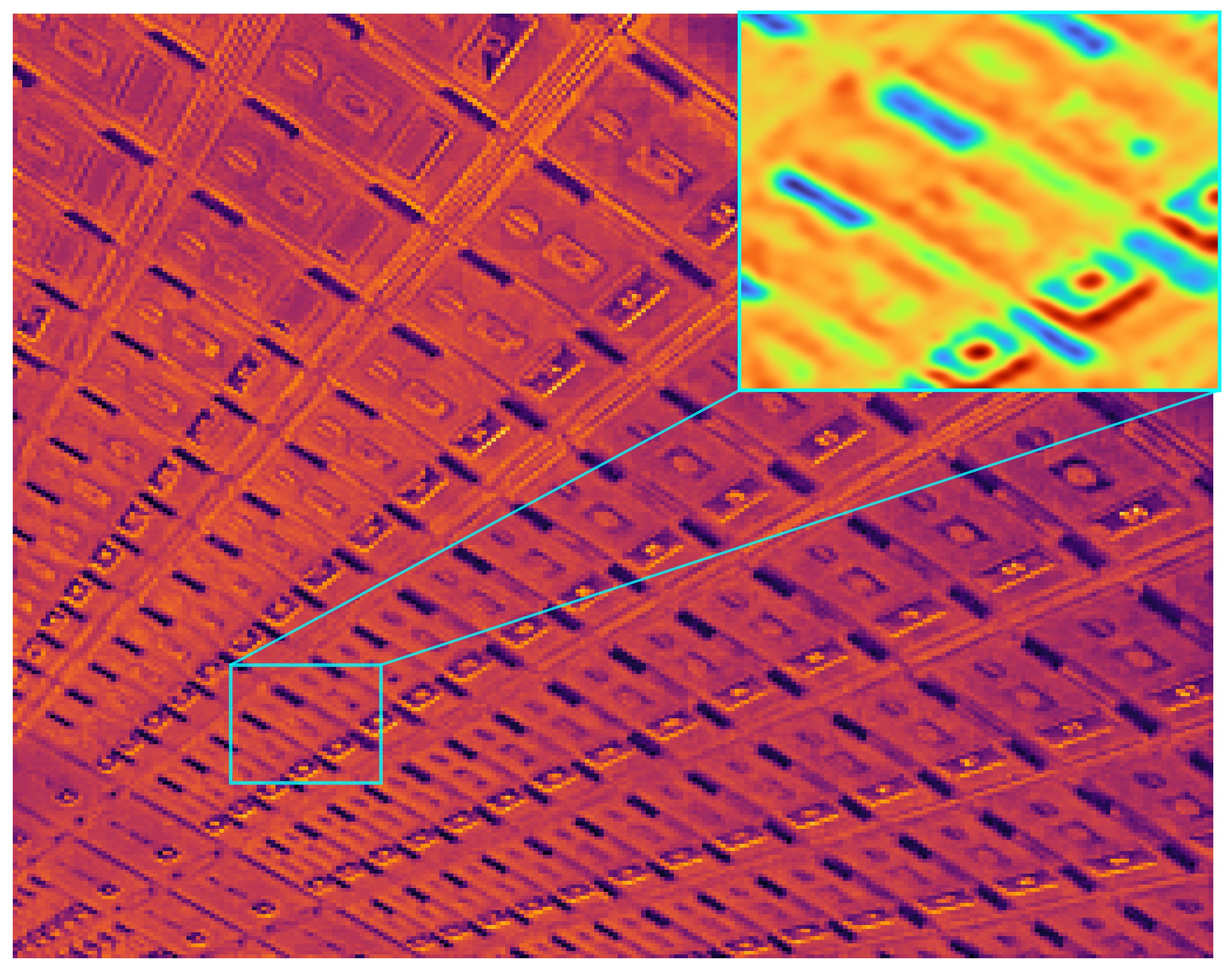}}
		\subfloat[W4A32 Attention Map]{\includegraphics[width=0.48\linewidth]{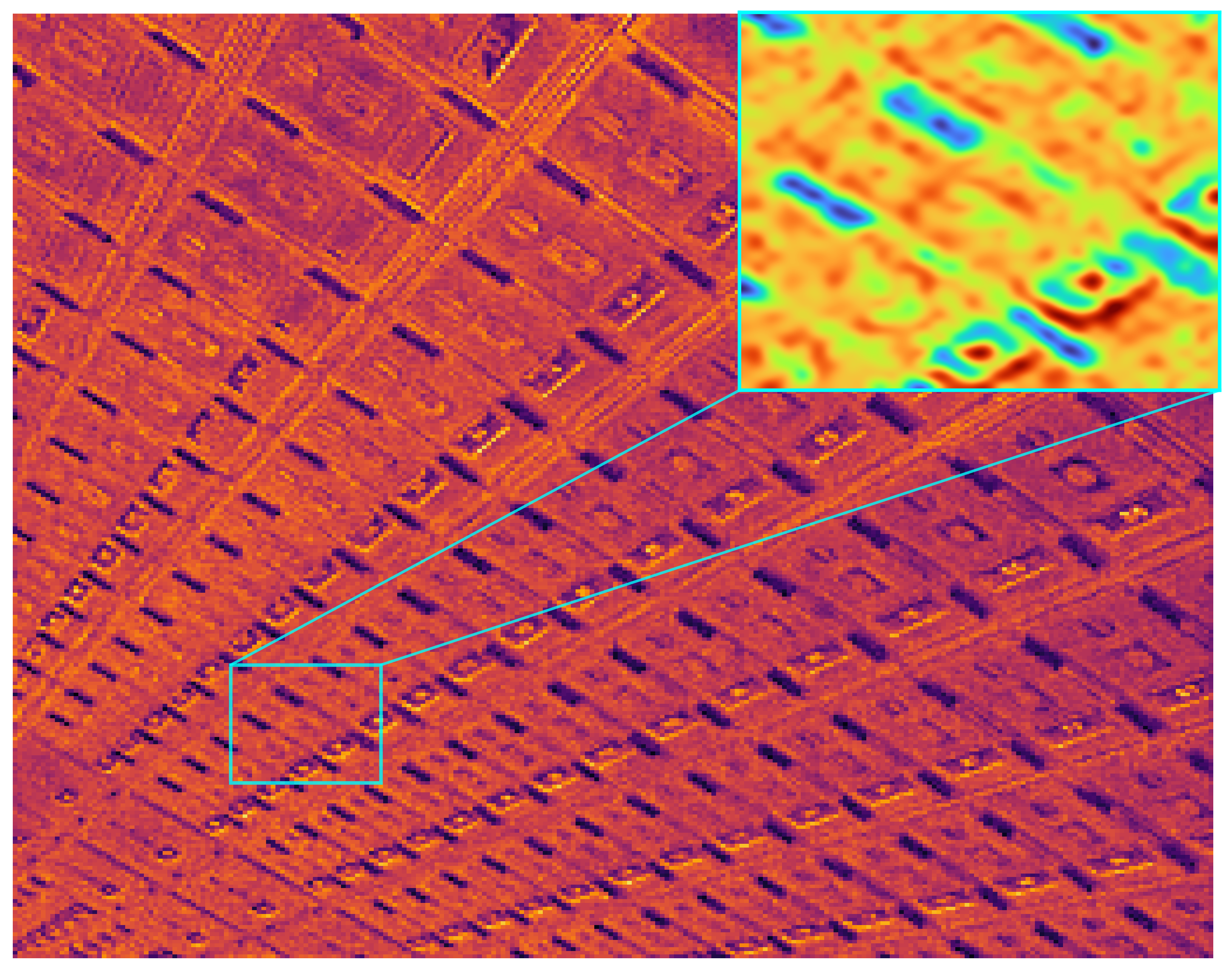}}
		\caption{\textbf{Layer-wise quantization sensitivity analysis in SwinIR.} Stacked bar chart shows the proportion of sensitivity to weight versus activation quantization across four layer types. Attention and GELU layers are highly sensitive to activation quantization (92.5\% and 93.6\%), while Shallow layers show balanced sensitivity. Attention maps compare (a) W32A4 and (b) W4A32 configurations, illustrating differential quantization impacts in attention layer.}
		\label{fig:msecompare}
	\end{figure}
	
	Despite these findings, most existing quantization methods either treat weight and activation quantization independently, as in DBDC~\cite{tu2023toward} and 2DQuant~\cite{liu20242dquant}, or focus on a single component, as in dynamic activation-only quantization~\cite{wang2025thinking}. These methods overlook the intrinsic  coupling between weights and activations. Classical PTQ techniques such as MinMax~\cite{jacob2018quantization} or Percentile~\cite{li2019fully} further exacerbate this issue by applying uniform clipping or scaling heuristics that ignore signal content related to structure and texture, leading to suboptimal quantization behavior.
	
	To address these challenges, we introduce HarmoQ, a harmonized post-training quantization framework designed specifically for super-resolution models. HarmoQ jointly minimizes the compound error caused by weight-activation interaction by incorporating structure-aware mechanisms. First, it introduces a structural residual calibration scheme that adjusts weights to proactively compensate for structural degradation arising from activation quantization. Second, it derives a closed-form optimal scaling factor that analytically balances the quantization difficulty across components. Third, it applies adaptive boundary refinement to maintain this balance consistently throughout optimization, even under aggressive quantization regimes.
	
	This work contributes the first systematic analysis of the asymmetric  effects of weight versus activation quantization in super-resolution. It also proposes a closed-form structural calibration technique to mitigate the loss of fine-grained details, and an analytically grounded scaling strategy for harmonizing quantization across components. Together, these advances lead to strong empirical gains in low-bit regimes, offering a robust and theoretically motivated solution for efficient super-resolution deployment.
	
	\section{Related Work}
	\paragraph{Image Super-resolution.}
	Image super-resolution aims to reconstruct high-resolution images from low-resolution inputs, with applications in medical imaging~\cite{TCJ2008SuperGreenspan,ICTSD2015SuperIsaac} and surveillance~\cite{ESP2010SuperZhang}.
	Early CNN approaches like SRCNN~\cite{ECCV2014LearningDong} pioneered deep learning for super-resolution. EDSR~\cite{lim2017enhanced} achieved state-of-the-art results using enhanced residual networks, while SRGAN~\cite{ledig2017photo} introduced adversarial training for photo-realistic results. RCAN~\cite{zhang2018image} and RDN~\cite{zhang2018residual} further improved performance through attention mechanisms and dense connections.
	Recent transformer-based methods have shown promising results. SwinIR~\cite{liang2021swinir} adapted Swin Transformer~\cite{liu2021swin} for image restoration, while Restormer~\cite{zamir2022restormer} proposed efficient transformer architectures for high-resolution restoration. CAT~\cite{chen2022cross} and DAT~\cite{chen2023dual} explore cross-aggregation and dual-aggregation mechanisms for enhanced feature learning.
	
	\paragraph{Model Quantization.}
	Model quantization reduces computational complexity by representing weights and activations with lower precision, enabling efficient deployment on resource-constrained devices.
	Early quantization-aware training methods include BinaryNets~\cite{courbariaux2016binarized}, which constrained weights to binary values, and DoReFa-Net~\cite{zhou2016dorefa}, which extended low-bitwidth quantization. PACT~\cite{choi2018pact} introduced parameterized clipping for quantized networks.
	Post-training quantization  offers practical advantages by quantizing pre-trained models without retraining. BRECQ~\cite{li2021brecq} achieved significant improvements through block-wise reconstruction, while recent work~\cite{ding2022towards} addresses PTQ challenges for vision transformers.
	For super-resolution specifically, PAMS~\cite{li2020pams} proposed parameterized max scale quantization, DAQ~\cite{hong2022daq} developed channel-wise distribution-aware methods, and CADYQ~\cite{hong2022cadyq} introduced content-aware dynamic quantization. Recent advances~\cite{tu2023toward} have made progress toward accurate post-training quantization for super-resolution models.
	
	\section{Problem Formulation}
	
	Consider a linear transformation in a super-resolution network:
	$y = Wx + b,$
	where $W \in \mathbb{R}^{m \times n}$ and $x \in \mathbb{R}^n$ denote the weights and input activation, respectively. Let $\tilde{W}$ and $\tilde{x}$ be their quantized counterparts, with corresponding quantization errors $\delta_W = \tilde{W} - W$ and $\delta_x = \tilde{x} - x$. The quantized forward pass expands to:
	\begin{equation}
		\tilde{y} = Wx + W\delta_x + \delta_W x + \delta_W\delta_x + b.
	\end{equation}
	Our analysis in Figure~\ref{fig:weightactivatecompare} reveals that weight and activation quantization exhibit complementary degradation patterns: weight quantization primarily impacts structural similarity (degrading SSIM), while activation quantization corrupts pixel-level accuracy (degrading PSNR). This metric differentiation makes independent optimization suboptimal, as existing methods~\cite{jacob2018quantization,tu2023toward} fail to minimize the critical compound error term $W\delta_x + \delta_W x$.
	
	We propose a harmonized quantization framework that jointly optimizes quantization parameters to preserve critical structural and pixel-level information. Our approach formulates the problem as:
	\begin{equation}
		\begin{aligned}
			\min_{s} \quad & \mathcal{L}_{\mathrm{total}}(s) = \mathbb{E}\left[\left\| W\delta_x(s) + \delta_W(s) x \right\|_2^2\right], \\
			\text{s.t.} \quad & \left| \mathrm{MSE}_x(s) - \mathrm{MSE}_w(s) \right| \leq \epsilon,
		\end{aligned}
	\end{equation}
	where $s$ is a scaling parameter shared across quantization modules, and $\epsilon$ enforces distortion symmetry. The optimal $s^*$ satisfies $\mathrm{MSE}_x(s^*) \approx \mathrm{MSE}_w(s^*)$, yielding minimal total error under balanced quantization.
	
	\begin{figure}[t]
		\centering
		\includegraphics[width=1\linewidth]{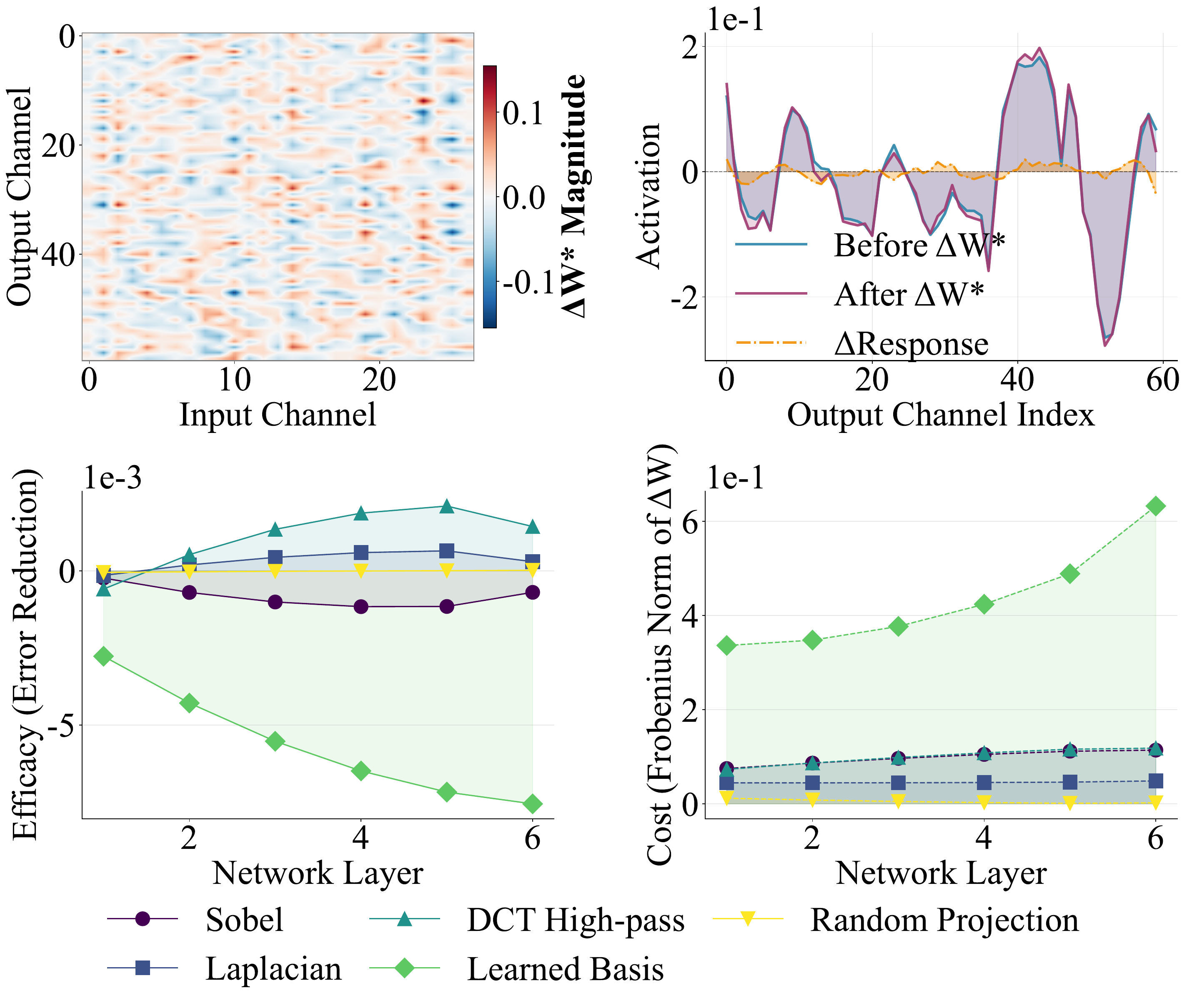}
%
\caption{\textbf{Structural Residual Calibration analysis in SwinIR.} 
	The top-left panel presents the optimal weight calibration $\Delta W^*$ heatmap using Laplacian filter, revealing structured channel relationships. 
	The top-right panel shows response comparison before and after $\Delta W^*$ calibration, demonstrating structure-aware modulation. 
	The bottom panel illustrates the efficacy across layers for different projection matrices. Laplacian filter achieves superior performance (see Table~\ref{tab:hf_projection_ablation} for quantitative comparison), while random projection fails due to lack of structure-aware design. Cost analysis via Frobenius norm of $\Delta W$ shows consistent magnitude across projection types.}
		\label{fig:case}
	\end{figure}
	
	\begin{figure*}[t]
		\centering
		\includegraphics[width=1\linewidth]{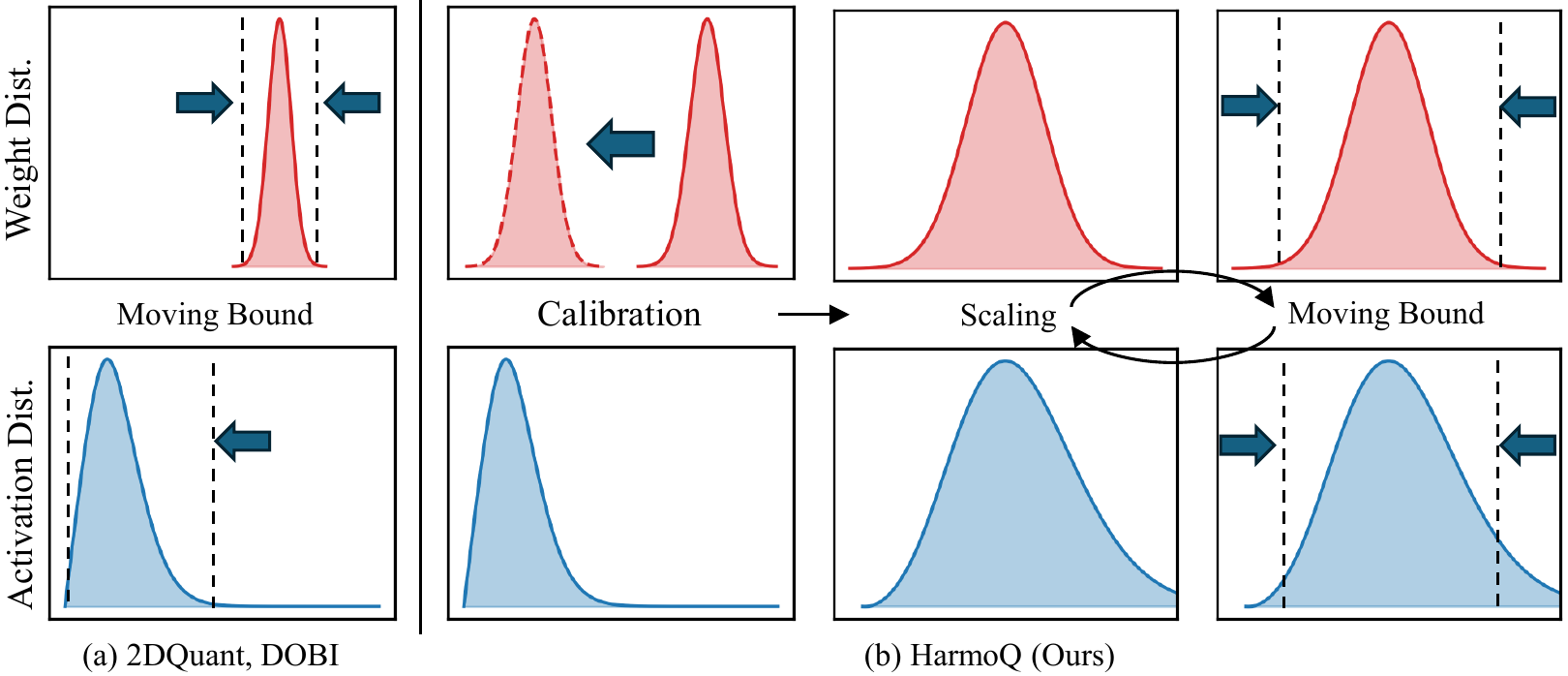}
		\caption{\textbf{Comparison of quantization optimization strategies.} (a) Existing methods (2DQuant, DOBI) independently optimize weight and activation quantization boundaries through separate calibration processes, leading to suboptimal parameter selection. (b) Our HarmoQ framework employs a unified three-step optimization: calibration for initial range estimation, harmonized scaling using optimal factor $s^*$ to balance quantization difficulty, and iterative moving bounds refinement to jointly minimize compound quantization errors. The arrows indicate the iterative optimization flow between scaling and boundary adjustment steps.}
		\label{fig:framework}
	\end{figure*}
	
	\section{Methodology}
	
	As illustrated in Figure~\ref{fig:framework}, traditional approaches like 2DQuant and DOBI (Figure~\ref{fig:framework}a) calibrate weight and activation quantization parameters separately, resulting in suboptimal boundary choices that neglect their interplay. In contrast, our HarmoQ framework (Figure~\ref{fig:framework}b) introduces a unified three-stage optimization strategy: (1) \emph{calibration} to set initial parameter ranges, (2) \emph{harmonized scaling} to balance quantization difficulty via an optimal scale factor $s^*$, and (3) \emph{adaptive boundary refinement} to jointly adjust clipping ranges, minimizing the compound error. This coordinated process preserves critical image fidelity components and achieves more robust super-resolution performance.
	
	\subsection{Step 1: Structural Residual Calibration}
	
	In quantized super-resolution networks, the compound error term $W\delta_x + 
	\delta_W x$ exhibits metric-specific degradation patterns. The core challenge 
	stems from weight quantization errors $\delta_W$, which systematically corrupt 
	the learned structural priors (e.g., for edges and textures) in $W$. This degradation is particularly 
	severe in SR tasks where preserving edge sharpness and texture details is critical.  Unlike weight quantization that introduces systematic biases affecting global image properties uniformly, activation quantization errors vary both spatially and temporally, making them particularly detrimental to the preservation of fine-grained visual details that are critical for super-resolution quality.
	
	We propose a Structural Residual Calibration (SRC) mechanism that 
	adjusts weight parameters to compensate for structural degradation introduced 
	by weight quantization $\delta_W$ itself.  This approach leverages the insight that weight parameters, with their distributed representation capacity, can be strategically modified to encode compensatory information that mitigates activation-induced loss of fine-grained detail before they manifest in the output.
	
	\paragraph{Subspace projection operator.}
	Let $H \in \mathbb{R}^{k \times d}$ be a linear projection matrix that extracts structural features like edges and textures from $d$-dimensional output features. This matrix can be constructed using fixed filters such as DCT high-pass masks or Laplacian edge filters bases derived from calibration data. 
	
	\paragraph{Structure-aware calibration objective.}
	We seek a calibration term $\delta_W$ that suppresses the projected residual energy while maintaining stability through regularization:
	\begin{equation}
		\min_{\delta_W} \mathbb{E} \left[
		\| H(W \delta_x + \delta_W x) \|_2^2
		\right] + \lambda \|\delta_W\|_F^2,
	\end{equation}
	where $\lambda > 0$ balances calibration strength and smoothness.
	
	\paragraph{Closed-form solution.}
	By expanding the objective and taking its derivative with respect to $\delta_W$, we obtain the following optimality condition:
	\begin{equation}
		\delta_W \left( H \mathbb{E}[x x^T] H^T + \lambda I \right)
		= - W \mathbb{E}[\delta_x x^T] H^T.
	\end{equation}
	Solving for $\delta_W$ yields the closed-form structure-aware calibration:
	\begin{equation}
		\delta_W^* = -W \mathbb{E}[\delta_x x^T] H^T
		\left( H \mathbb{E}[x x^T] H^T + \lambda I \right)^{-1}
	\end{equation}
	The complete derivation of this result is provided in \textit{Appendix}.
	
	\paragraph{Projection analysis.} To validate our structural residual calibration, we analyze the optimal weight calibration $\Delta W^*$ across different projection matrices. The top-left panel of Figure~\ref{fig:case} shows the heatmap of $\Delta W^*$ derived via a Laplacian filter, revealing structured, non-random patterns that target specific channel relationships, confirming the structure-specific corrective capability of our method.
	The top-right panel illustrates how $\Delta W^*$ modulates activation responses while preserving key signal characteristics. The smooth transitions in the corrected profile support our structure-aware calibration principle and demonstrate reduced structural distortion.
	The bottom panel further shows consistent error reduction across layers for various structural projection filters. DCT high-pass and learned basis filters yield superior results, especially in deeper layers where quantization errors accumulate, while random projections perform poorly, underscoring the importance of a structure-aware design.
	
	\subsection{Step 2: Harmonized Scale Optimization}
	\label{sec:scaling}
	
	Following the structural calibration, we address the fundamental challenge of \emph{non-uniform quantization difficulty} across weights and activations due to their distinct statistical properties. We seek a harmonizing scale factor $s$ that balances quantization difficulty between components.
	
	\paragraph{Scale harmonization objective.}
	Given clipping boundaries $\theta = \{\alpha_x, \beta_x, \alpha_w, \beta_w\}$, we seek a scale factor $s$ that ensures equal quantization difficulty: $\mathrm{MSE}_x(s) = \mathrm{MSE}_w(s).$
	In SR networks, the activation range is further modulated by an input scaling factor $s$, which affects MSE differently for activations and weights:
	\begin{align*}
		\mathrm{MSE}_x(s) = \frac{(\beta_x - \alpha_x)^2}{12 s^2 (2^{b_x} - 1)^2}, \ 
		\mathrm{MSE}_w(s) = \frac{(\beta_w - \alpha_w)^2 s^2}{12 (2^{b_w} - 1)^2}.
	\end{align*}
	Setting $\mathrm{MSE}_x(s) = \mathrm{MSE}_w(s)$ and solving for $s$ gives the closed-form optimal scale:
	\begin{equation}
		\label{eq:scale_closed}
		s^* = \sqrt{\frac{(\beta_x - \alpha_x)(2^{b_w} - 1)}{(\beta_w - \alpha_w)(2^{b_x} - 1)}}.
	\end{equation}
	The detailed derivation of the optimal scale factor is given in \textit{Appendix}. This balances the difficulty of quantizing weights and activations, particularly in structure-sensitive regimes, ensuring that neither component dominates the quantization error.
	
	\subsection{Step 3: Adaptive Boundary Refinement}
	\label{sec:moving_bounds}
	
	With the harmonized scale $s^*$ established, we optimize the quantization boundaries $\theta$ to minimize the total reconstruction error while maintaining the balanced quantization difficulty achieved in Step 2.
	
	\paragraph{Boundary optimization objective.}
	We minimize the total quantization error by updating the quantization boundaries $\theta$:
	$$
	\mathcal{L}_{\mathrm{total}}(s^*, \theta) = \mathbb{E} \left[ \left\| W \cdot \delta_x(s^*, \theta) + \delta_W(s^*, \theta) \cdot x \right\|_2^2 \right].
	$$
	The gradient with respect to any boundary parameter $\theta_i \in \theta$ is:
	$$
	\frac{\partial \mathcal{L}_{\mathrm{total}}}{\partial \theta_i} 
	= \mathbb{E} \left[ 2 \left( W \delta_x + \delta_W x \right)^T \cdot \left( W \frac{\partial \delta_x}{\partial \theta_i} + \frac{\partial \delta_W}{\partial \theta_i} x \right) \right].
	$$
	
	\paragraph{Gradient computation for clipping boundaries.}
	Since $\delta_x$ and $\delta_W$ are functions of the quantizer's clipping range, their derivatives w.r.t. $\alpha$ and $\beta$ depend on the quantization function. For a symmetric uniform quantizer $Q(z) = \mathrm{clip}(z; \alpha, \beta)$ with step size $\Delta = \frac{\beta - \alpha}{2^b - 1}$, we can compute:
	$$
	\frac{\partial \delta}{\partial \alpha} = \frac{\partial}{\partial \alpha} (Q(z) - z), \qquad
	\frac{\partial \delta}{\partial \beta} = \frac{\partial}{\partial \beta} (Q(z) - z),
	$$
	which are differentiable almost everywhere except at clipping points. The complete gradient derivation for boundary optimization is presented in \textit{Appendix}.

	\paragraph{Boundary update rule.}
	The boundaries are updated using projected gradient descent:
	\begin{equation}
		\theta^{(t+1)} = \mathrm{Proj}_\Omega \left( \theta^{(t)} - \eta \nabla_\theta \mathcal{L}_{\mathrm{total}}(s^*, \theta^{(t)}) \right),
	\end{equation}
	where $\Omega$ is a feasible set ensuring $\alpha < \beta$ and $\eta$ is the learning rate.
	
	\subsection{Iterative Refinement Process}
	
	The unified framework alternates between three coordinated operations:
	\begin{enumerate}
		\item \textbf{Step 1 - Structural Residual Calibration:} Applies closed-form structural calibration $\delta W_\ell^* = -W_\ell \Sigma_{\delta x} H^T (H \Sigma_{xx} H^T + \lambda I)^{-1}$ to minimize structural distortion.
		\item \textbf{Step 2 - Harmonized Scale Optimization:} Computes optimal scale factor $s^* = \sqrt{\frac{(\beta_x - \alpha_x)(2^{b_w} - 1)}{(\beta_w - \alpha_w)(2^{b_x} - 1)}}$ to maintain equal quantization difficulty.
		\item \textbf{Step 3 - Adaptive Boundary Refinement:} Updates clipping boundaries $\theta$ via gradient descent while enforcing the constraint $|\text{MSE}_x(s^*,\theta) - \text{MSE}_w(s^*,\theta)| \leq \epsilon$.
	\end{enumerate}
	
	\begin{table*}[t!]
		\centering
		\renewcommand{\arraystretch}{1.1}
		\resizebox{\textwidth}{!}{
			\begin{small}
				\renewcommand{\multirowsetup}{\centering}
				\setlength{\tabcolsep}{1.3pt}
				\begin{tabular}{lcccccccccc}
					\toprule[0.15em]
					\rowcolor[HTML]{F0F0F8} 
					\cellcolor[HTML]{F0F0F8} & \cellcolor[HTML]{F0F0F8} & \cellcolor[HTML]{F0F0F8}Set5 ($\times 2$) & \cellcolor[HTML]{F0F0F8}Set14 ($\times 2$) & \cellcolor[HTML]{F0F0F8}BSD100 ($\times 2$) & \cellcolor[HTML]{F0F0F8}Urban100 ($\times 2$) & \cellcolor[HTML]{F0F0F8}Set5 ($\times 4$) & \cellcolor[HTML]{F0F0F8}Set14 ($\times 4$) & \cellcolor[HTML]{F0F0F8}BSD100 ($\times 4$) & \cellcolor[HTML]{F0F0F8}Urban100 ($\times 4$) \\
					\rowcolor[HTML]{F0F0F8} 
					\multirow{-2}{*}{\cellcolor[HTML]{F0F0F8}Method} & \multirow{-2}{*}{\cellcolor[HTML]{F0F0F8}Bit} & \cellcolor[HTML]{F0F0F8}PSNR/SSIM & PSNR/SSIM & PSNR/SSIM & PSNR/SSIM & PSNR/SSIM & PSNR/SSIM & PSNR/SSIM & PSNR/SSIM \\ 
					\midrule[0.15em]
					\multicolumn{1}{l|}{SwinIR-light} & \multicolumn{1}{c|}{32}   & 38.15/0.9611 & 33.86/0.9206 & 32.31/0.9012 & 32.76/0.9340 & 32.45/0.8976 & 28.77/0.7858 & 27.69/0.7406 & 26.48/0.7980    \\
					\multicolumn{1}{l|}{Bicubic} & \multicolumn{1}{c|}{32}   & 32.25/0.9118 & 29.25/0.8406 & 28.68/0.8104 & 25.96/0.8088 & 27.56/0.7896 & 25.51/0.6820 & 25.54/0.6466 & 22.68/0.6352  \\ 
					\midrule
					\midrule
					\multicolumn{1}{l|}{MinMax} & \multicolumn{1}{c|}{2}       & 33.88/0.9185 & 30.81/0.8748 & 29.99/0.8535 & 27.48/0.8501 & 23.96/0.4950 & 22.92/0.4407 & 22.70/0.3943 & 21.16/0.4053 \\
					\multicolumn{1}{l|}{Percentile} & \multicolumn{1}{c|}{2}   & 30.82/0.8016 & 28.80/0.7616 & 27.95/0.7232 & 26.30/0.7378 & 23.03/0.4772 & 22.12/0.4059 & 21.83/0.3816 & 20.45/0.3951  \\
					\multicolumn{1}{l|}{DBDC+Pac} & \multicolumn{1}{c|}{2}   & 34.55/0.9386 & 31.12/0.8912 & 30.27/0.8706 & 27.63/0.8649 & 25.01/0.5554 & 23.82/0.4995 & 23.64/0.4544 & 21.84/0.4631 \\
					\multicolumn{1}{l|}{DOBI} & \multicolumn{1}{c|}{2}         & 35.25/0.9361 & 31.72/0.8917 & 30.62/0.8699 & 28.52/0.8727 & 28.82/0.7699 & 26.46/0.6804 & 25.97/0.6319 & 23.67/0.6407 \\
					\multicolumn{1}{l|}{Granular-DQ} & \multicolumn{1}{c|}{2} & 35.85/0.9485 & 31.85/0.8995 & 30.80/0.8795 & 28.45/0.8800 & 29.35/0.8355 & 26.75/0.7305 & 26.35/0.6910 & 23.70/0.6895 \\
					\multicolumn{1}{l|}{2DQuant} & \multicolumn{1}{c|}{2}       & 36.00/0.9497 & 31.98/0.9012 & 30.91/0.8810 & 28.62/0.8819 & 29.53/0.8372 & 26.86/0.7322 & 26.46/0.6927 & 23.84/0.6912 \\
					\multicolumn{1}{l|}{HarmoQ} & \multicolumn{1}{c|}{2} & {\color[HTML]{FD6864} \bf 36.46/0.9515 } & {\color[HTML]{FD6864} \bf 32.24/0.9037 } & {\color[HTML]{FD6864} \bf 31.12/0.8836 } & {\color[HTML]{FD6864} \bf 29.18/0.8893 } & {\color[HTML]{FD6864} \bf 30.23/0.8543 } & {\color[HTML]{FD6864} \bf 27.29/0.7446 } & {\color[HTML]{FD6864} \bf 26.70/0.7031 } & {\color[HTML]{FD6864} \bf 24.27/0.7123 } \\ 
					
					\midrule
					\multicolumn{1}{l|}{MinMax} & \multicolumn{1}{c|}{3}       & 28.19/0.6961 & 26.40/0.6478 & 25.83/0.6225 & 25.19/0.6773 & 19.41/0.3385 & 18.35/0.2549 & 18.79/0.2434 & 17.88/0.2825 \\
					\multicolumn{1}{l|}{Percentile} & \multicolumn{1}{c|}{3}   & 34.37/0.9170 & 31.04/0.8646 & 29.82/0.8339 & 28.25/0.8417 & 27.55/0.7270 & 25.15/0.6043 & 24.45/0.5333 & 22.80/0.5833  \\
					\multicolumn{1}{l|}{DBDC+Pac} & \multicolumn{1}{c|}{3}   & 35.07/0.9350 & 31.52/0.8873 & 30.47/0.8665 & 28.44/0.8709 & 27.91/0.7250 & 25.86/0.6451 & 25.65/0.6239 & 23.45/0.6249 \\
					\multicolumn{1}{l|}{DOBI} & \multicolumn{1}{c|}{3}         & 36.37/0.9496 & 32.33/0.9041 & 31.12/0.8836 & 29.65/0.8967 & 29.59/0.8237 & 26.87/0.7156 & 26.24/0.6735 & 24.17/0.6880 \\
					\multicolumn{1}{l|}{Granular-DQ} & \multicolumn{1}{c|}{3} & 37.20/0.9555 & 32.75/0.9095 & 31.50/0.8900 & 30.30/0.9070 & 30.75/0.8690 & 27.65/0.7560 & 26.90/0.7115 & 24.70/0.7340 \\
					\multicolumn{1}{l|}{2DQuant} & \multicolumn{1}{c|}{3} & 37.32/0.9567 & 32.85/0.9106 & 31.60/0.8911 & 30.45/0.9086 & 30.90/0.8704 & 27.75/0.7571 & 26.99/0.7126 & 24.85/0.7355 \\ 
					\multicolumn{1}{l|}{HarmoQ} & \multicolumn{1}{c|}{3} & {\color[HTML]{FD6864} \bf 37.98/0.9602 } & {\color[HTML]{FD6864} \bf 34.04/0.9220 } & {\color[HTML]{FD6864} \bf 32.23/0.9004 } & {\color[HTML]{FD6864} \bf 33.01/0.9368 } &  {\color[HTML]{FD6864} \bf 31.09/0.8740 } & {\color[HTML]{FD6864} \bf 27.91/0.7618 } & {\color[HTML]{FD6864} \bf 27.11/0.7186 } & {\color[HTML]{FD6864} \bf 25.31/0.7548 } \\ 
					
					\midrule[0.15em]
					
					\multicolumn{1}{l|}{HAT-S} & \multicolumn{1}{c|}{32} &  38.58/0.9628 & 34.70/0.9261&32.59/0.9050 & 34.31/0.9459  & 32.92/0.9047 & 29.15/0.7958  & 27.97/0.7505  & 27.87/0.8346           \\
					\multicolumn{1}{l|}{Bicubic} & \multicolumn{1}{c|}{32}   & 32.25/0.9118 & 29.25/0.8406 & 28.68/0.8104 & 25.96/0.8088 & 27.56/0.7896 & 25.51/0.6820 & 25.54/0.6466 & 22.68/0.6352  \\
					\midrule
					\midrule
					\multicolumn{1}{l|}{MinMax} & \multicolumn{1}{c|}{2}       & 34.63/0.9263 & 31.23/0.8798 & 30.59/0.8595 & 30.95/0.8777 & 25.95/0.5746 & 24.16/0.4888 & 23.78/0.4483 & 22.37/0.4734 \\
					\multicolumn{1}{l|}{Percentile} & \multicolumn{1}{c|}{2}   & 32.06/0.8756 & 29.89/0.8118 & 28.97/0.7842 & 28.68/0.8288 & 24.18/0.5398 & 22.89/0.4507 & 22.34/0.4162 & 21.01/0.4346  \\
					\multicolumn{1}{l|}{DBDC+Pac} & \multicolumn{1}{c|}{2}   & 35.89/0.9457 & 32.29/0.9001 & 31.28/0.8777 & 30.61/0.8934 & 26.97/0.6285 & 25.16/0.5493 & 24.56/0.5044 & 23.18/0.5127 \\
					\multicolumn{1}{l|}{DOBI} & \multicolumn{1}{c|}{2}         & 36.63/0.9413 & 32.93/0.9008 & 31.69/0.8769 & 31.68/0.9004 & 31.04/0.8297 & 27.51/0.7104 & 26.89/0.6859 & 25.13/0.7089 \\
					\multicolumn{1}{l|}{Granular-DQ} & \multicolumn{1}{c|}{2} & 36.70/0.9525 & 32.55/0.9070 & 31.30/0.8870 & 30.05/0.9045 & 30.30/0.8495 & 27.30/0.7405 & 26.80/0.7005 & 25.00/0.7410 \\
					\multicolumn{1}{l|}{2DQuant} & \multicolumn{1}{c|}{2}       & 36.81/0.9538 & 32.66/0.9085 & 31.42/0.8883 & 30.21/0.9061 & 30.48/0.8513 & 27.42/0.7423 & 26.89/0.7019 & 25.14/0.7425 \\
					\multicolumn{1}{l|}{HarmoQ} & \multicolumn{1}{c|}{2} & {\color[HTML]{FD6864} \bf 38.34/0.9617 } & {\color[HTML]{FD6864} \bf 34.35/0.9246 } & {\color[HTML]{FD6864} \bf 32.14/0.8961 } & {\color[HTML]{FD6864} \bf 33.87/0.9434 } &  {\color[HTML]{FD6864} \bf 31.18/0.8749 } & {\color[HTML]{FD6864} \bf 27.95/0.7626 } & {\color[HTML]{FD6864} \bf 27.12/0.7191 } & {\color[HTML]{FD6864} \bf 25.33/0.7560 } \\ 
					\midrule
					\multicolumn{1}{l|}{MinMax} & \multicolumn{1}{c|}{3}       & 30.86/0.7707 & 27.61/0.6957 & 26.07/0.6607 & 27.44/0.7567 & 21.77/0.4146 & 20.02/0.3287 & 19.77/0.3061 & 19.96/0.3764 \\
					\multicolumn{1}{l|}{Percentile} & \multicolumn{1}{c|}{3}   & 35.66/0.9283 & 32.17/0.8789 & 30.42/0.8453 & 31.37/0.8938 & 29.36/0.7855 & 26.13/0.6542 & 25.19/0.5998 & 25.16/0.7056  \\
					\multicolumn{1}{l|}{DBDC+Pac} & \multicolumn{1}{c|}{3}   & 36.41/0.9436 & 32.63/0.8977 & 30.99/0.8719 & 31.70/0.9018 & 29.82/0.7951 & 26.55/0.6769 & 25.71/0.6264 & 25.47/0.7178 \\
					\multicolumn{1}{l|}{DOBI} & \multicolumn{1}{c|}{3}         & 37.77/0.9558 & 33.49/0.9156 & 31.70/0.8901 & 33.09/0.9274 & 31.38/0.8586 & 27.83/0.7466 & 26.94/0.7085 & 26.89/0.8002 \\
					\multicolumn{1}{l|}{Granular-DQ} & \multicolumn{1}{c|}{3} & 37.00/0.9520 & 33.15/0.9115 & 31.40/0.8860 & 32.65/0.9225 & 30.90/0.8520 & 27.45/0.7405 & 26.60/0.7030 & 26.40/0.7935 \\
					\multicolumn{1}{l|}{2DQuant} & \multicolumn{1}{c|}{3}       & 37.11/0.9532 & 33.26/0.9128 & 31.52/0.8872 & 32.78/0.9239 & 31.02/0.8534 & 27.57/0.7421 & 26.71/0.7043 & 26.54/0.7948 \\ 
					\multicolumn{1}{l|}{HarmoQ} & \multicolumn{1}{c|}{3} & {\color[HTML]{FD6864} \bf 38.35/0.9618 } & {\color[HTML]{FD6864} \bf 34.42/0.9250 } & {\color[HTML]{FD6864} \bf 32.36/0.9021 } & {\color[HTML]{FD6864} \bf 34.21/0.9441 } &  {\color[HTML]{FD6864} \bf 31.90/0.8821 } & {\color[HTML]{FD6864} \bf 28.34/0.7700 } & {\color[HTML]{FD6864} \bf 27.41/0.7268 } & {\color[HTML]{FD6864} \bf 26.71/0.8106 } \\ 
					\bottomrule[0.15em]
				\end{tabular}%
			\end{small}
		}
		\caption{Quantitative comparison of HarmoQ against SOTA quantization methods on standard SR benchmarks.}		\label{tab:quantitative-comparison-complete}
	\end{table*}
	
	\section{Experiments}
	
	\paragraph{Baseline Methods.}
	We compare HarmoQ against several categories of quantization methods. Traditional approaches include MinMax~\cite{jacob2018quantization} and Percentile~\cite{li2019fully} quantization, which use activation extrema and percentiles respectively to determine quantization boundaries. General post-training quantization methods include BRECQ~\cite{li2021brecq}, which pushes quantization limits through block-wise reconstruction, and DOBI, a data-free method that independently optimizes weight and activation quantization.
	For super-resolution specific methods, we evaluate against PAMS~\cite{li2020pams} (parameterized max scale quantization), DBDC+Pac~\cite{tu2023toward} (accurate post-training quantization for image super-resolution), 2DQuant~\cite{liu20242dquant} (low-bit post-training quantization), and Dynamic Granularity~\cite{wang2025thinking} (activation-only quantization strategy). 
	
	\begin{figure*}[t]
		\centering
		\scriptsize
		\tabcolsep=2pt
		\begin{tabular}{cccccccc}
			\multirow{-6.1}{*}{\includegraphics[width=0.35\linewidth]{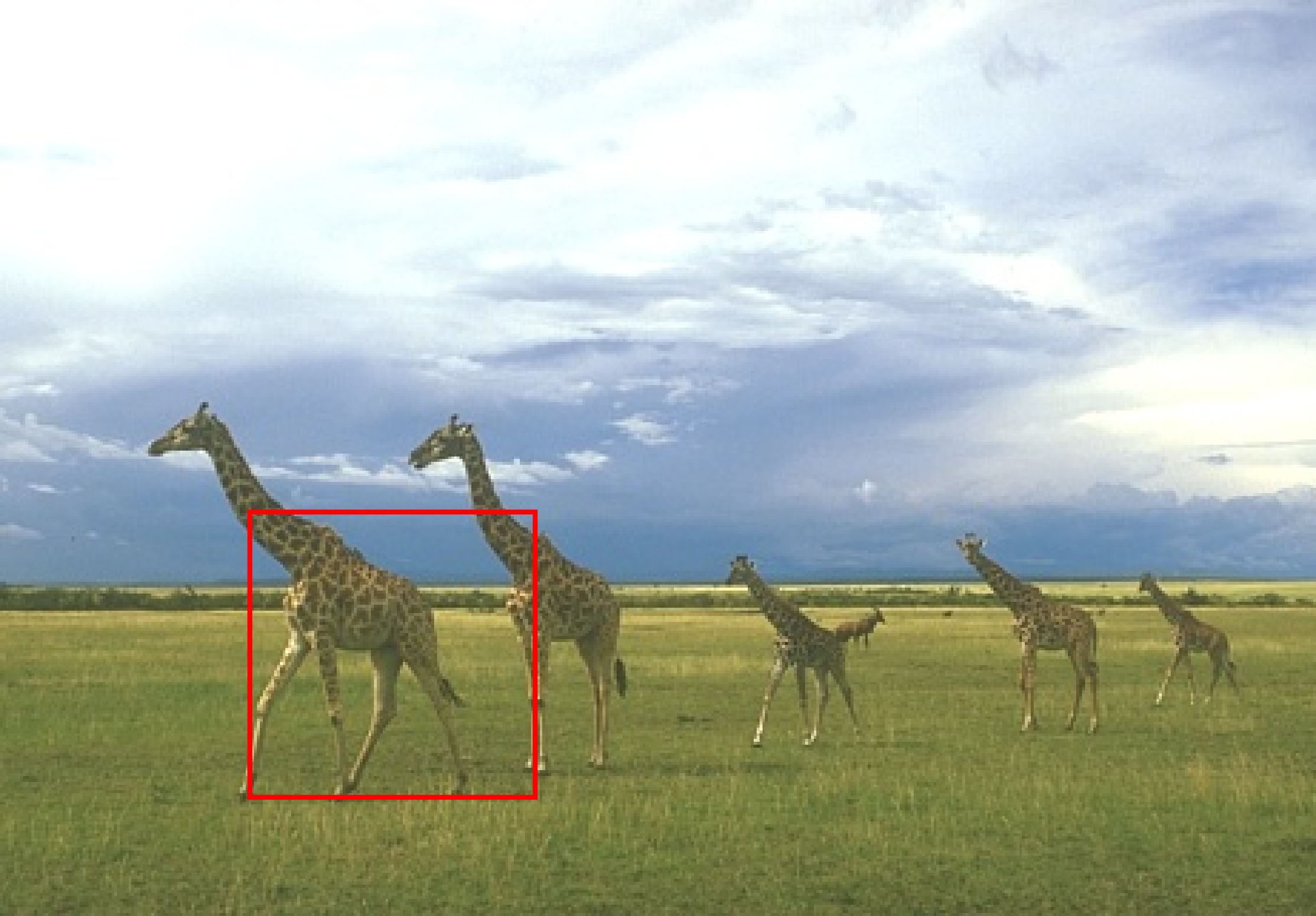}}
			& \includegraphics[width=0.115\linewidth]{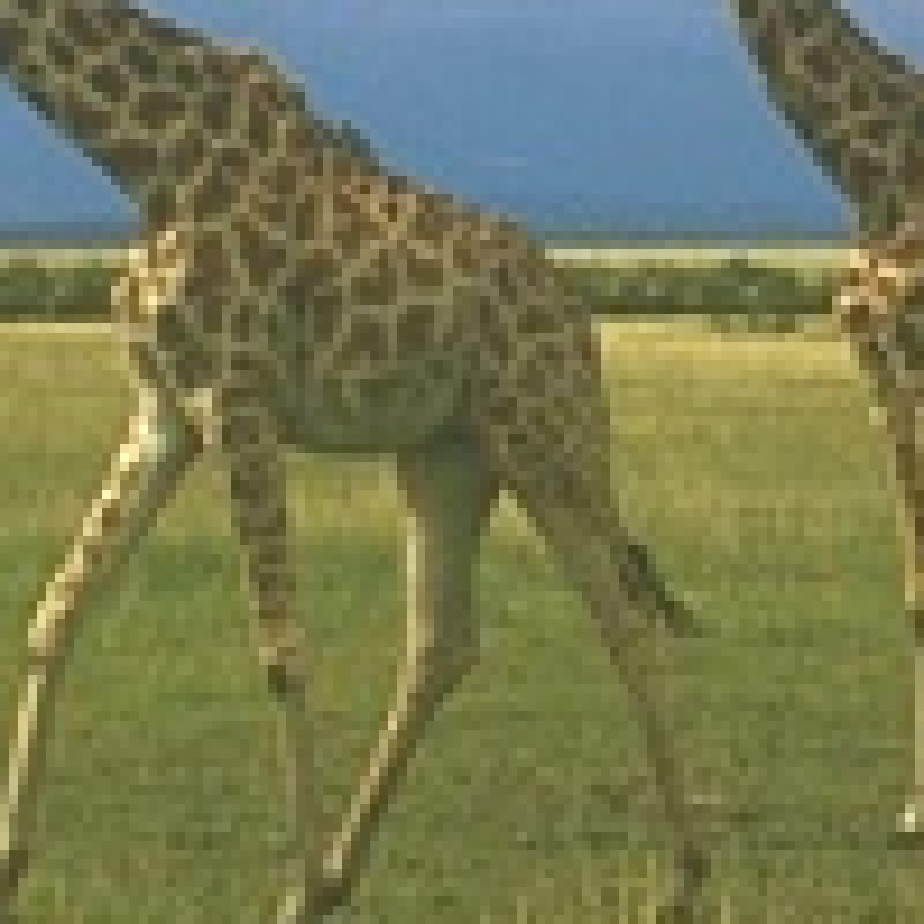}
			& \includegraphics[width=0.115\linewidth]{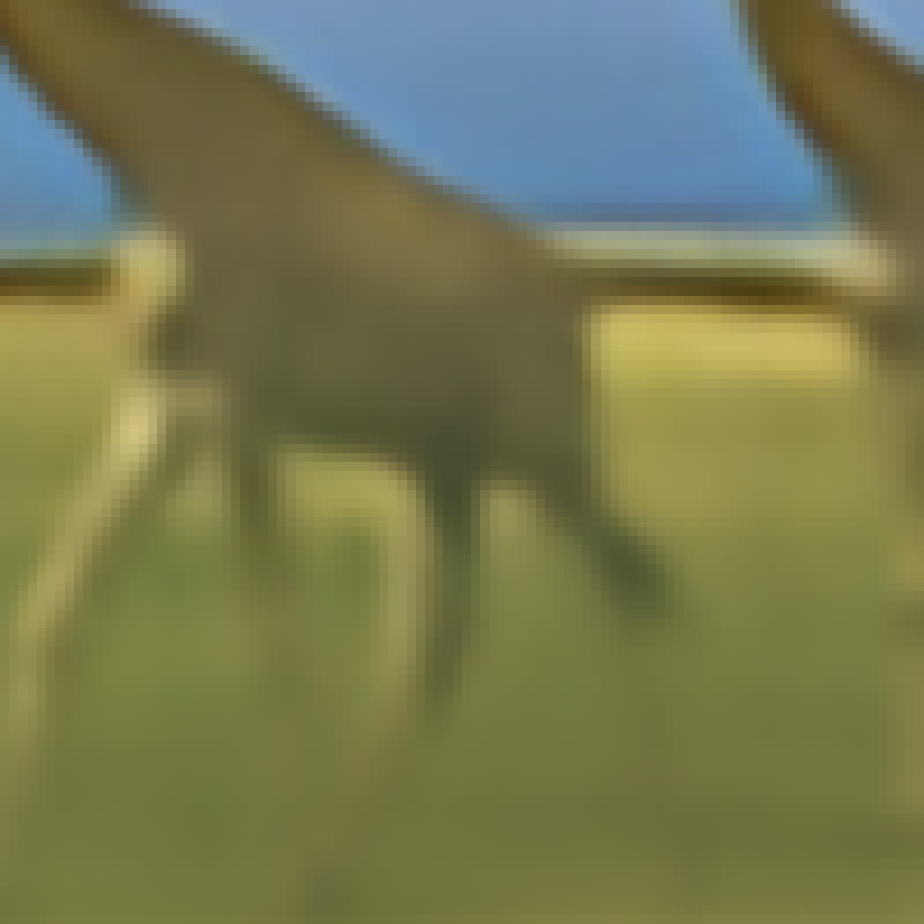}
			& \includegraphics[width=0.115\linewidth]{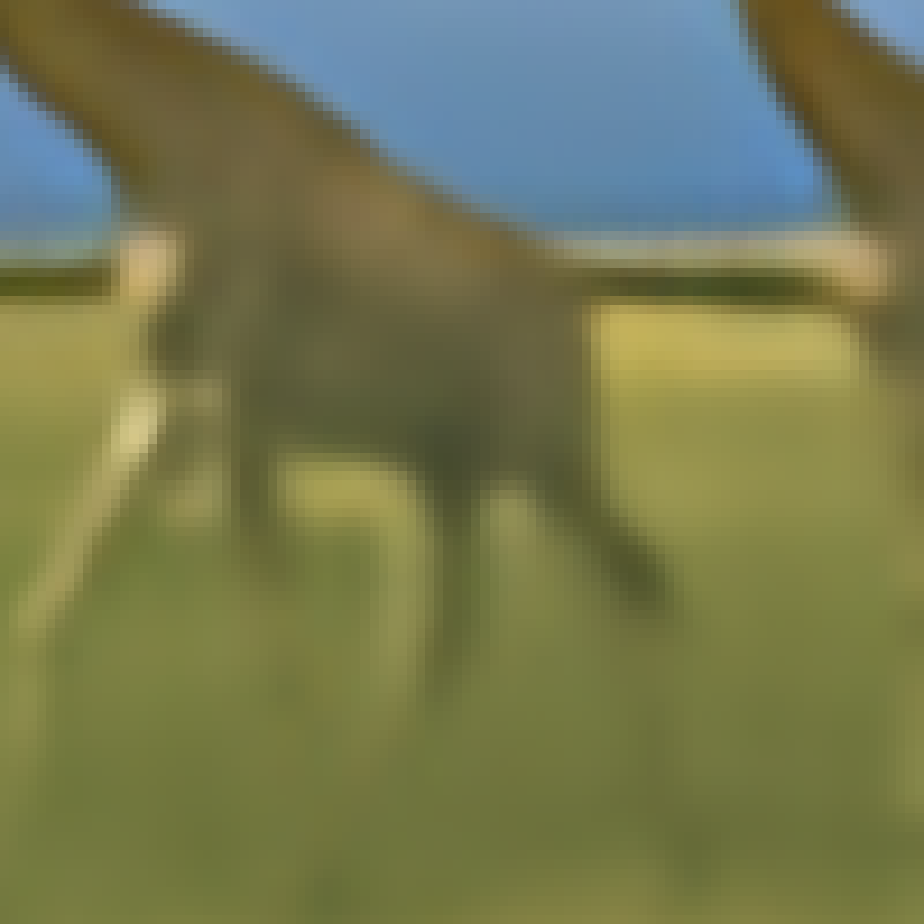}
			& \includegraphics[width=0.115\linewidth]{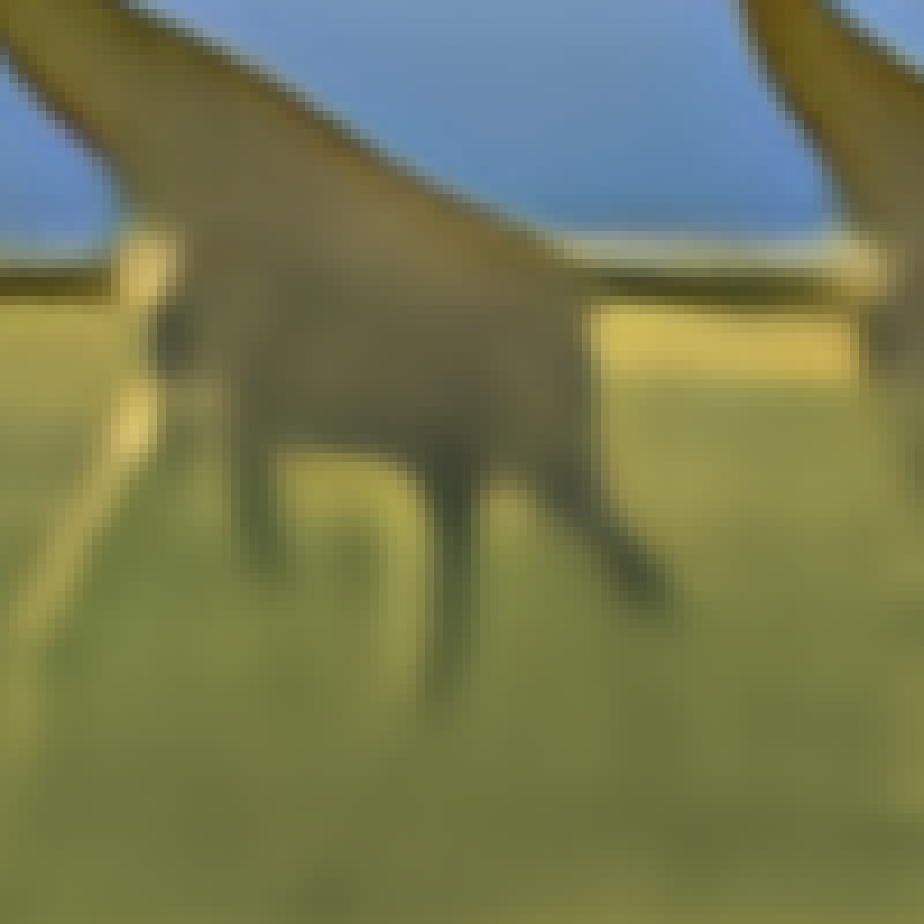}
			& \includegraphics[width=0.115\linewidth]{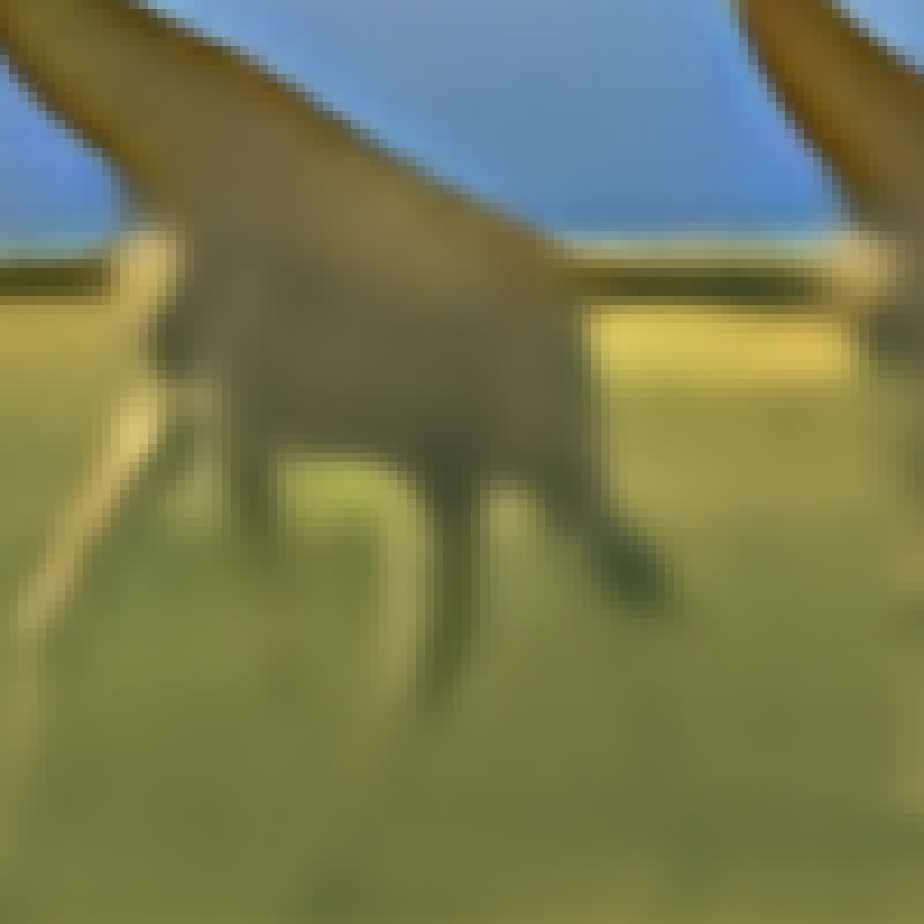}
			\\
			& GT Patch   & SwinIR-Light (3bit) & MinMax & 2DQuant & DBDC 
			\\
			& \includegraphics[width=0.115\linewidth]{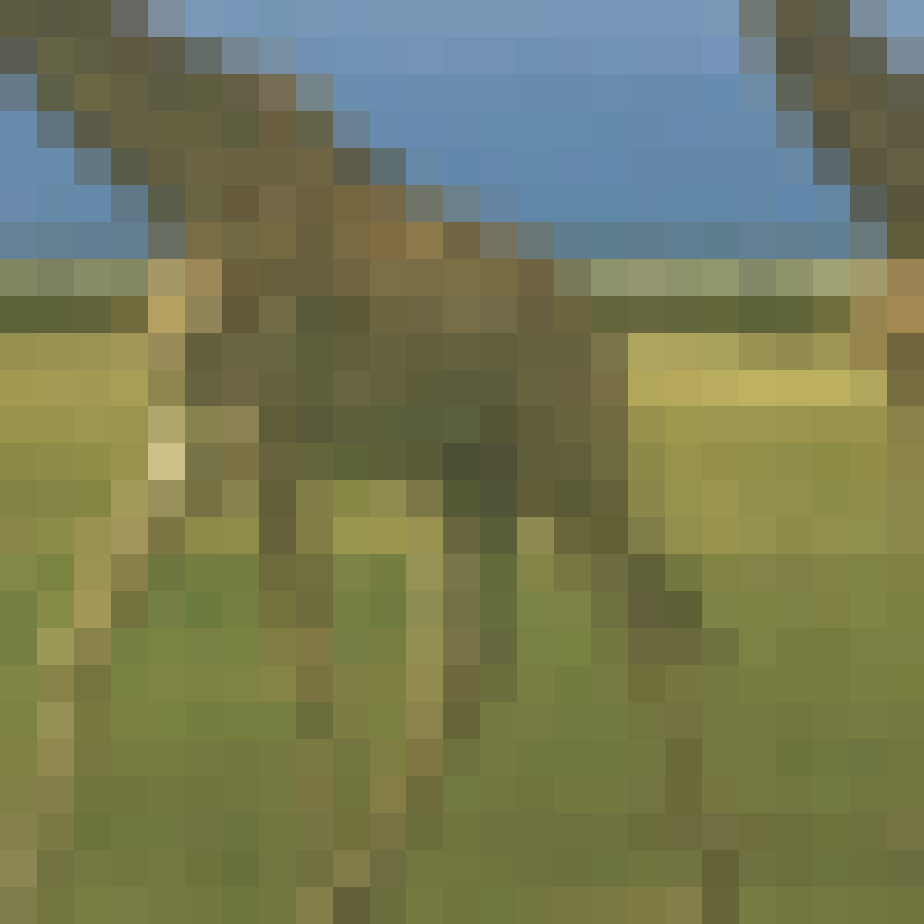}
			& \includegraphics[width=0.115\linewidth]{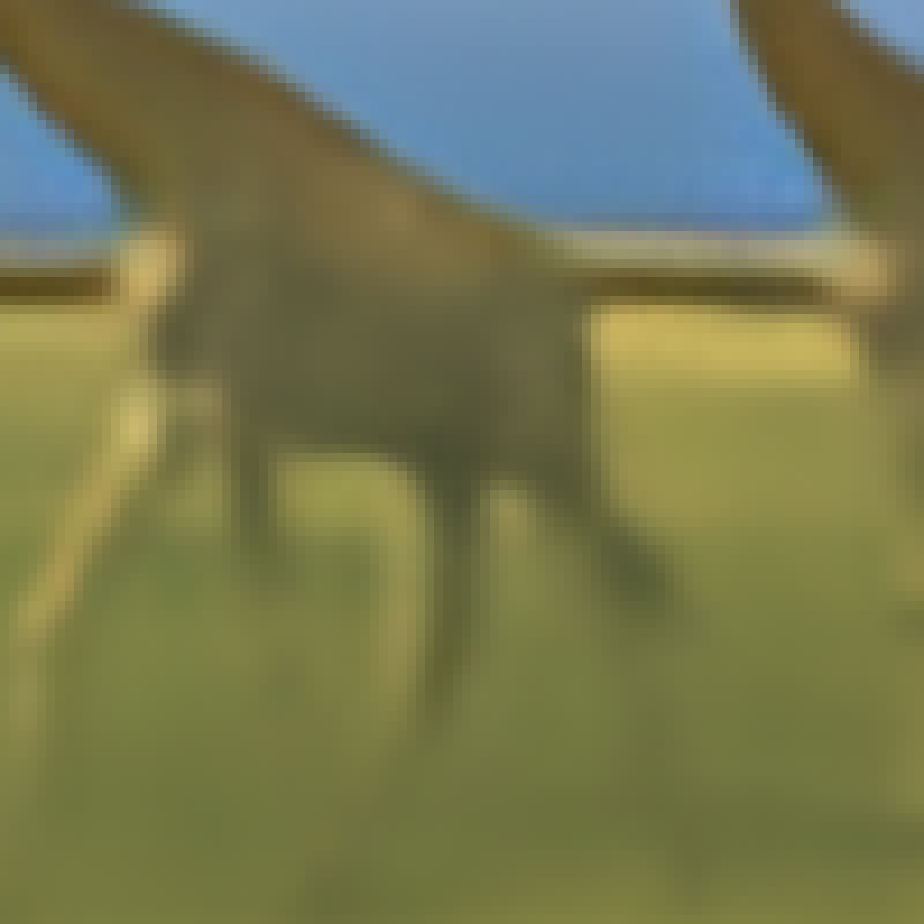}
			& \includegraphics[width=0.115\linewidth]{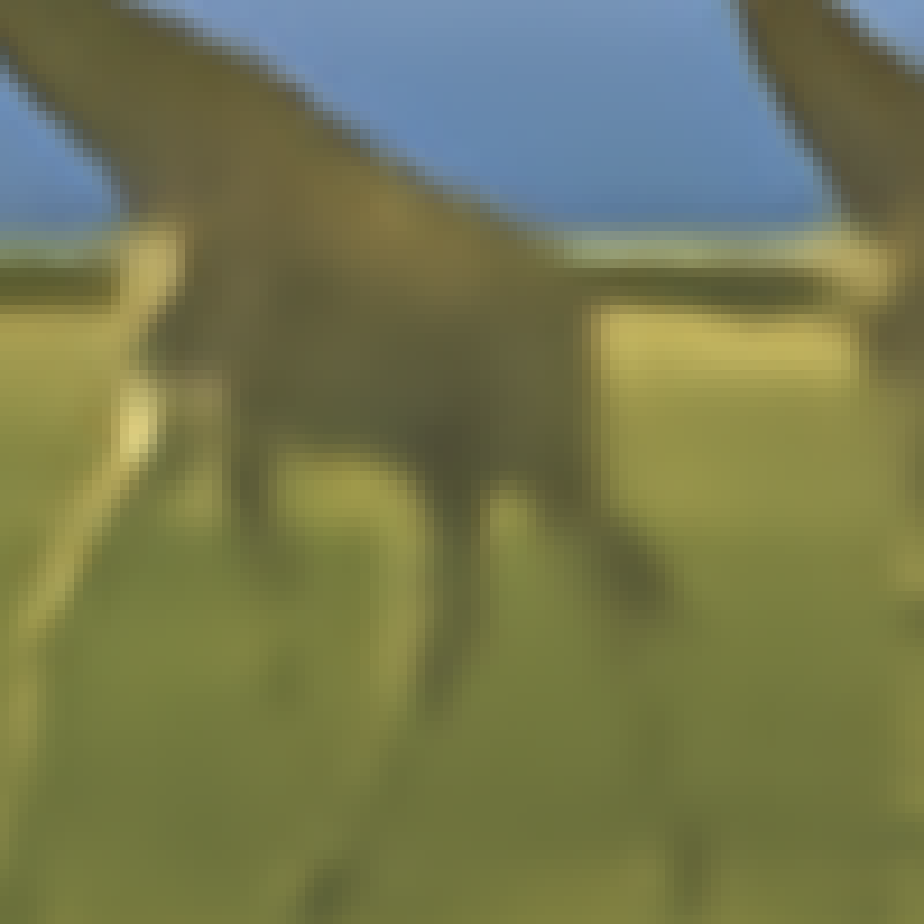}
			& \includegraphics[width=0.115\linewidth]{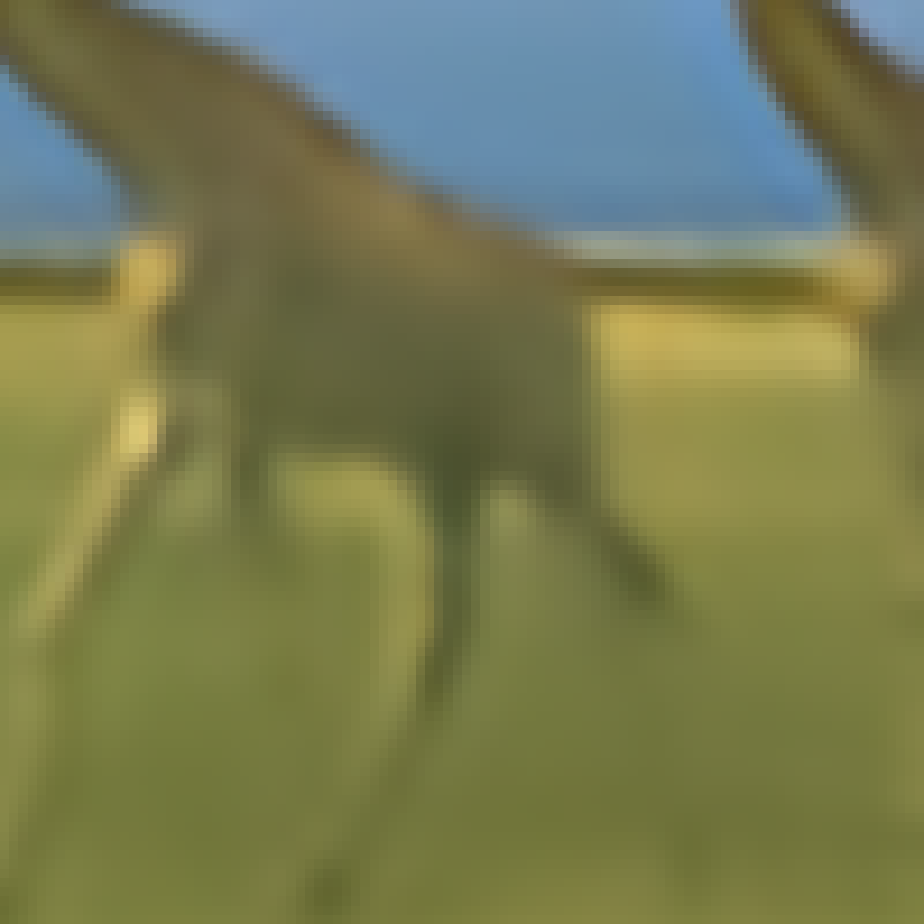}
			& \includegraphics[width=0.115\linewidth]{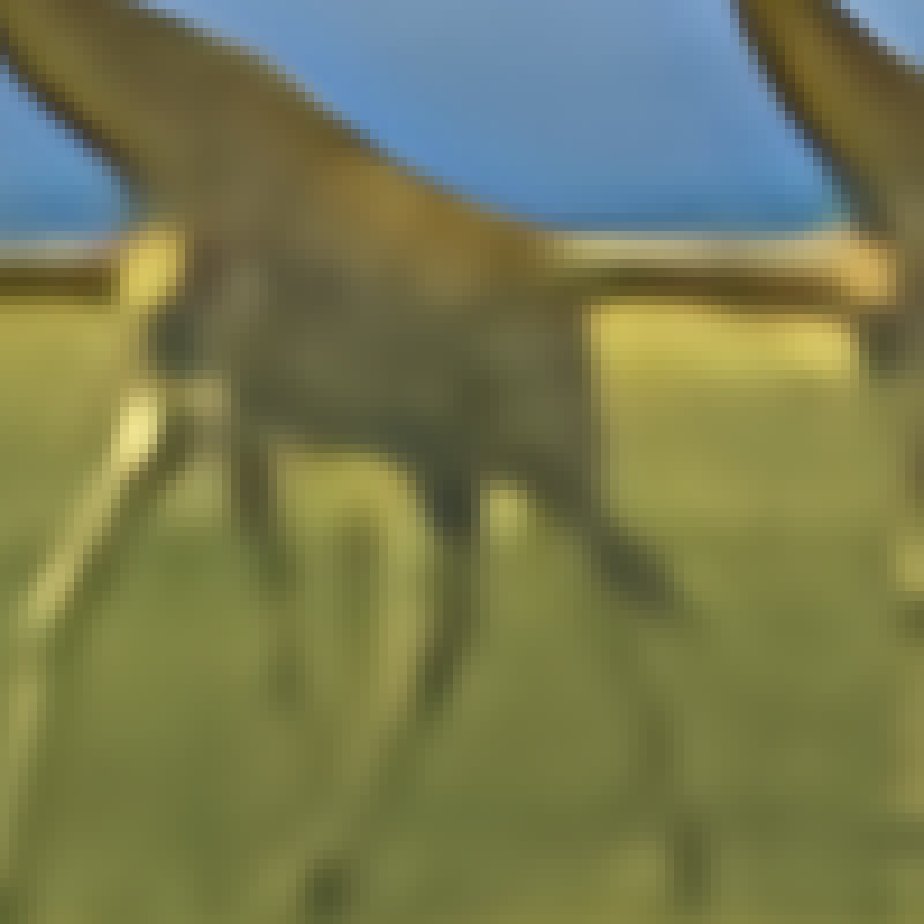}
			\\
			\multirow{-1.5}{*}{\small Reference Image} & LR Patch & DOBI & Granular-DQ &HarmoQ (w.o. SRC)      & HarmoQ (Ours)
		\end{tabular}
		\begin{tabular}{cccccccc}
			\multirow{-6.1}{*}{\includegraphics[width=0.35\linewidth]{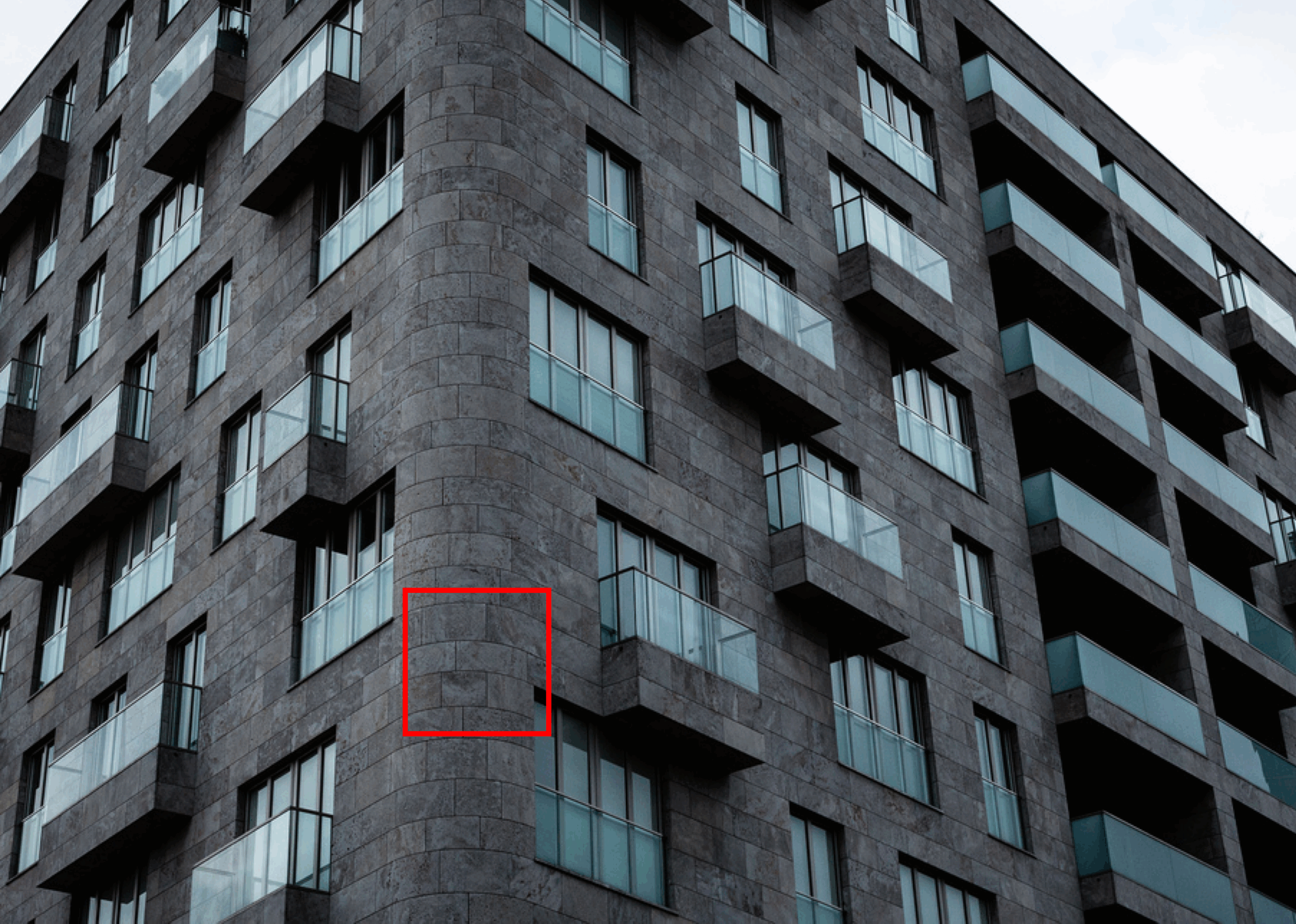}}
			& \includegraphics[width=0.115\linewidth]{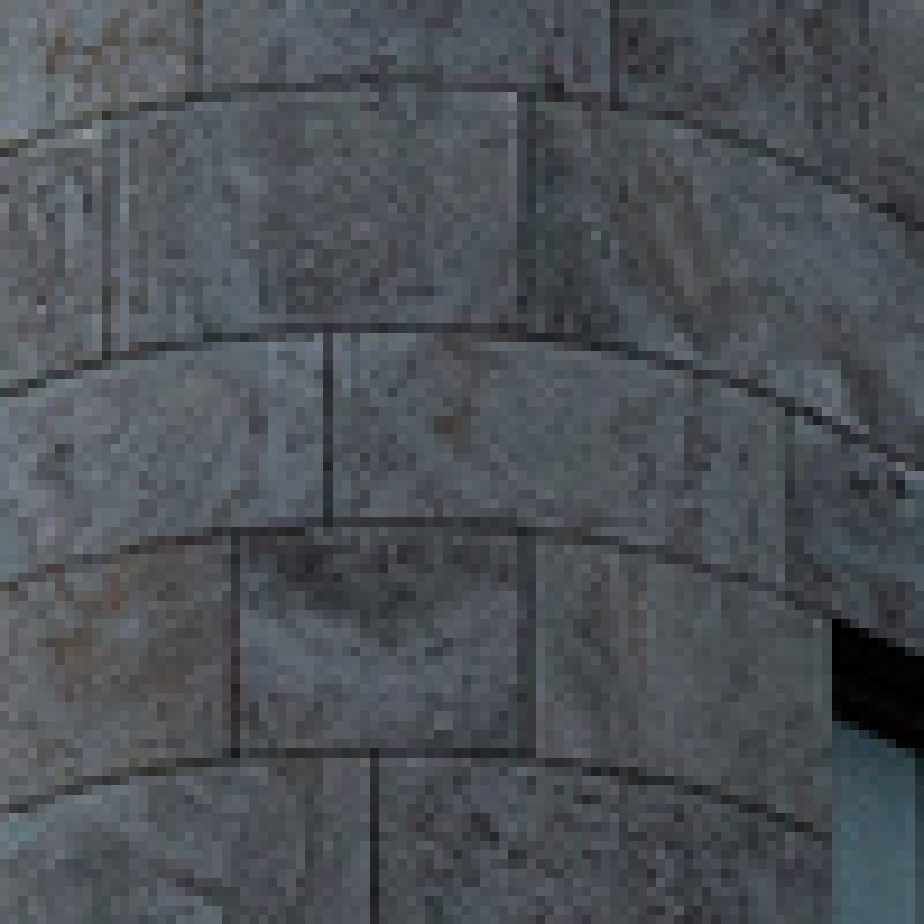}
			& \includegraphics[width=0.115\linewidth]{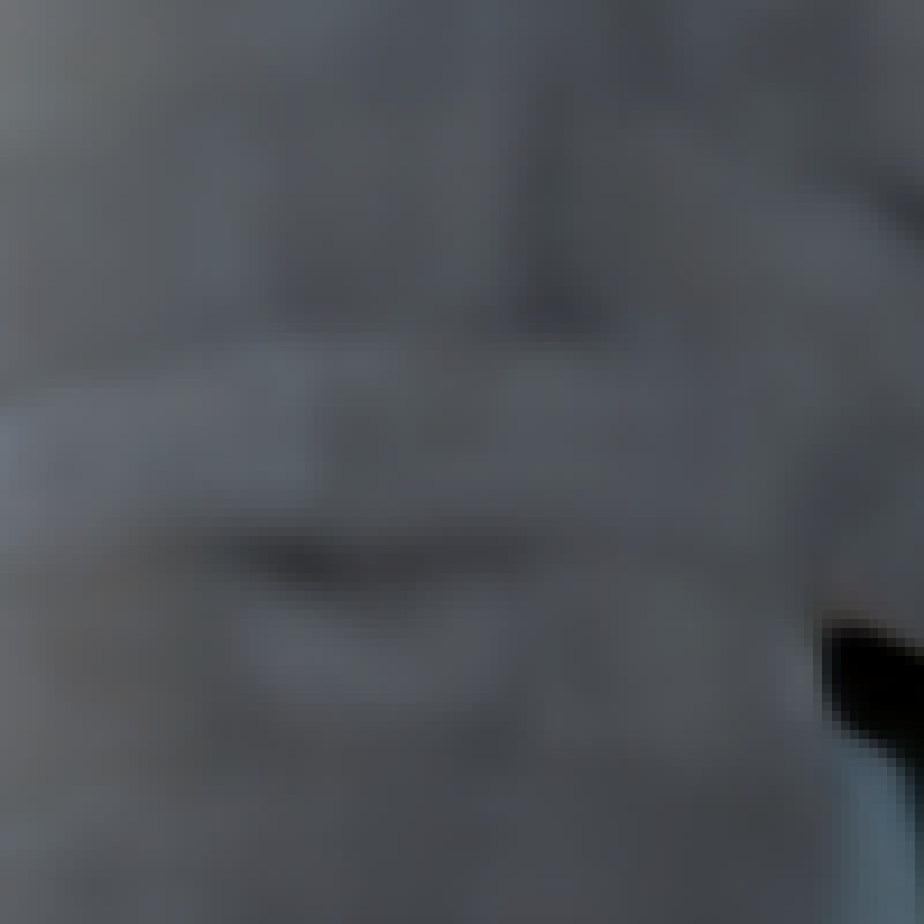}
			& \includegraphics[width=0.115\linewidth]{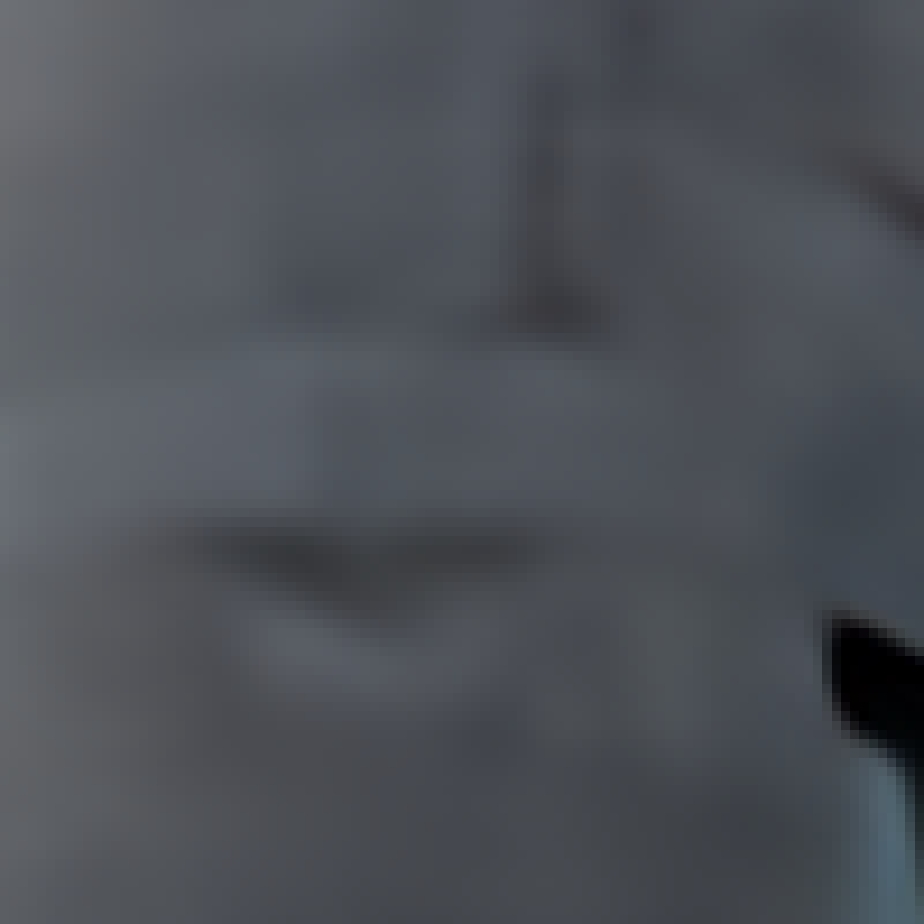}
			& \includegraphics[width=0.115\linewidth]{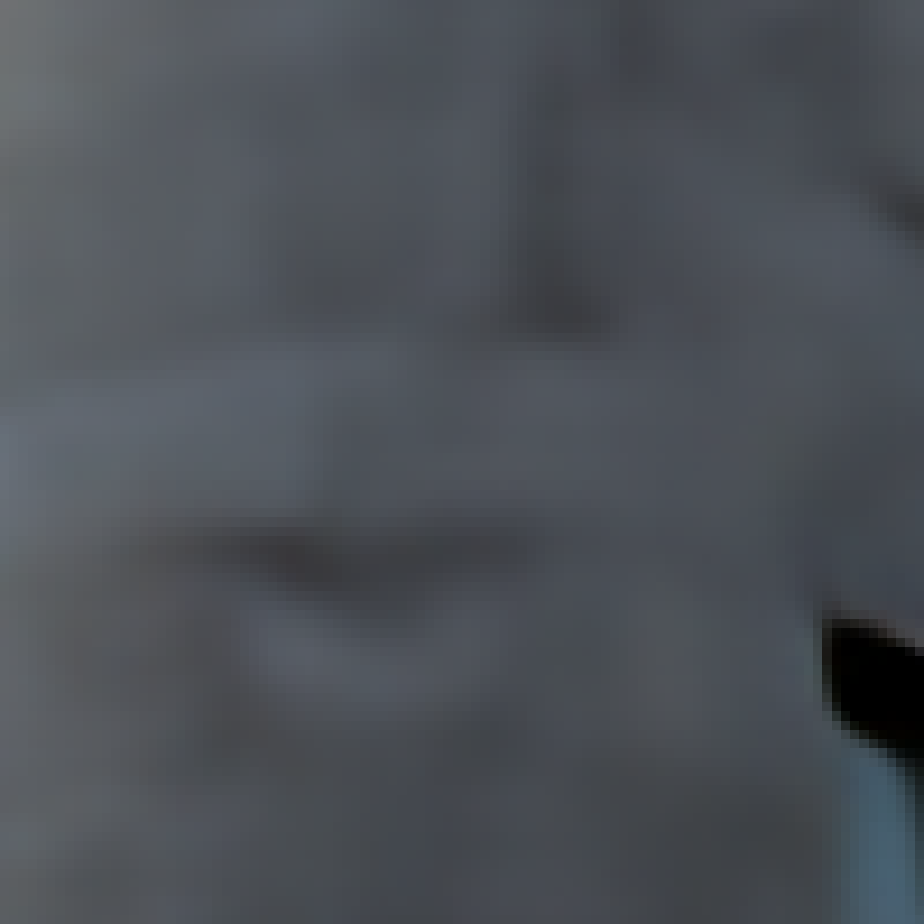}
			& \includegraphics[width=0.115\linewidth]{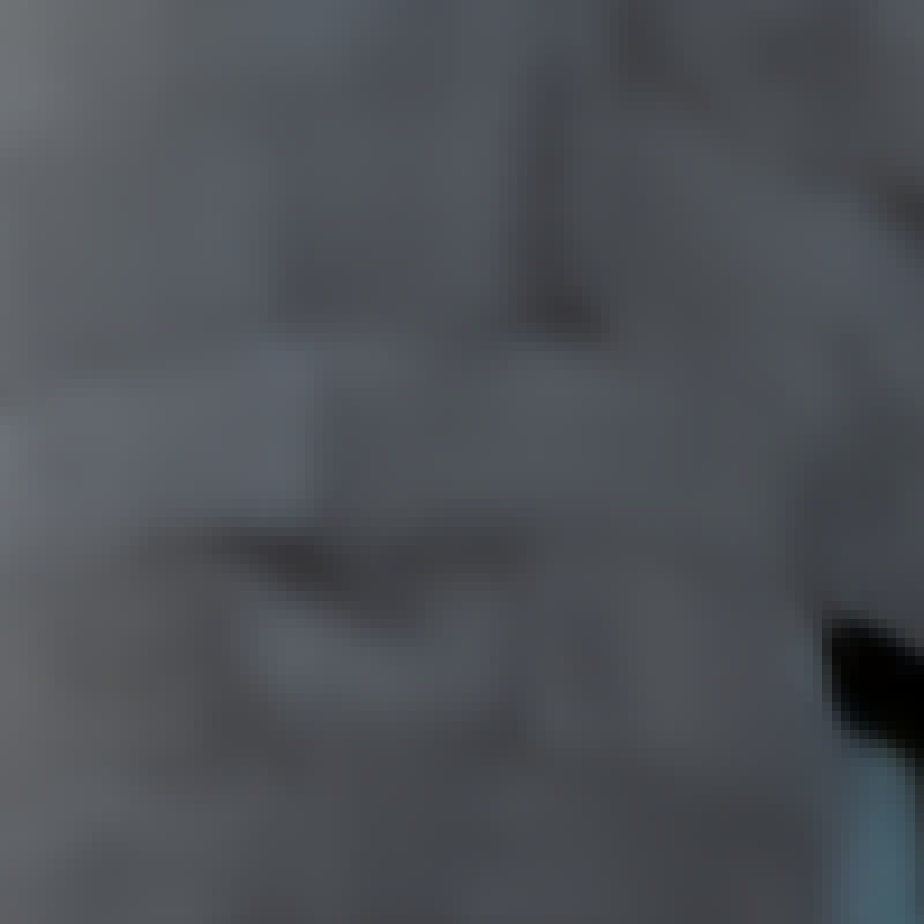}
			\\
			& GT Patch   & HAT-S (3bit) & MinMax & 2DQuant & DBDC 
			\\
			& \includegraphics[width=0.115\linewidth]{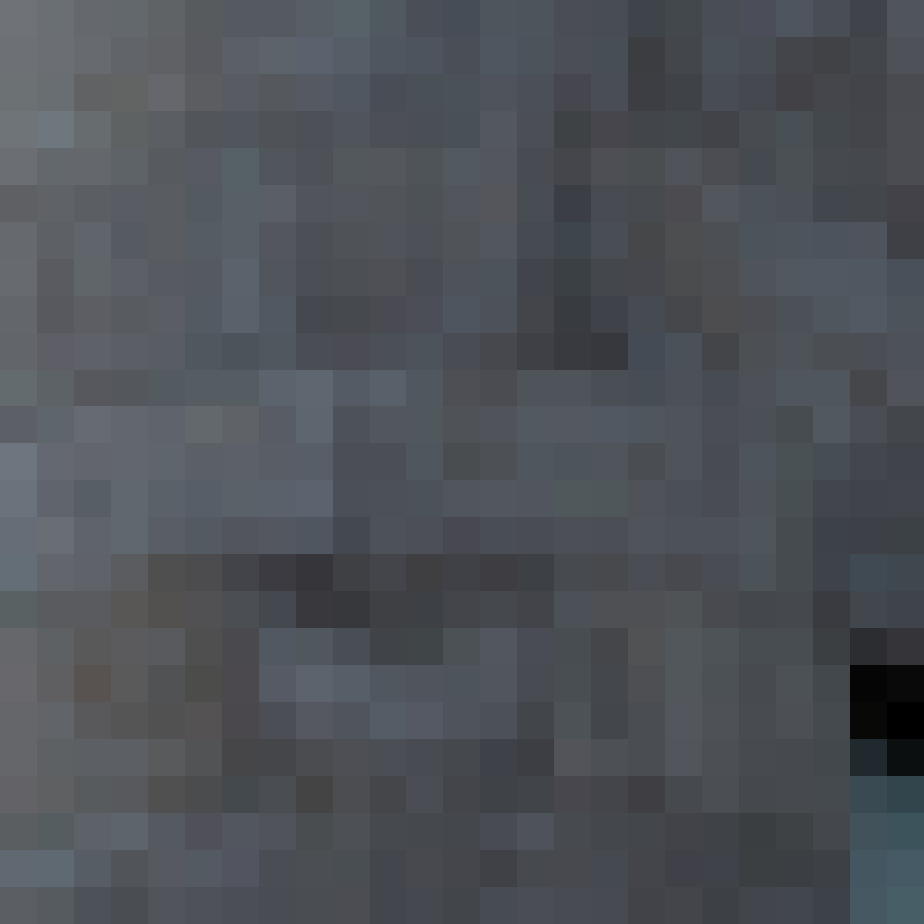}
			& \includegraphics[width=0.115\linewidth]{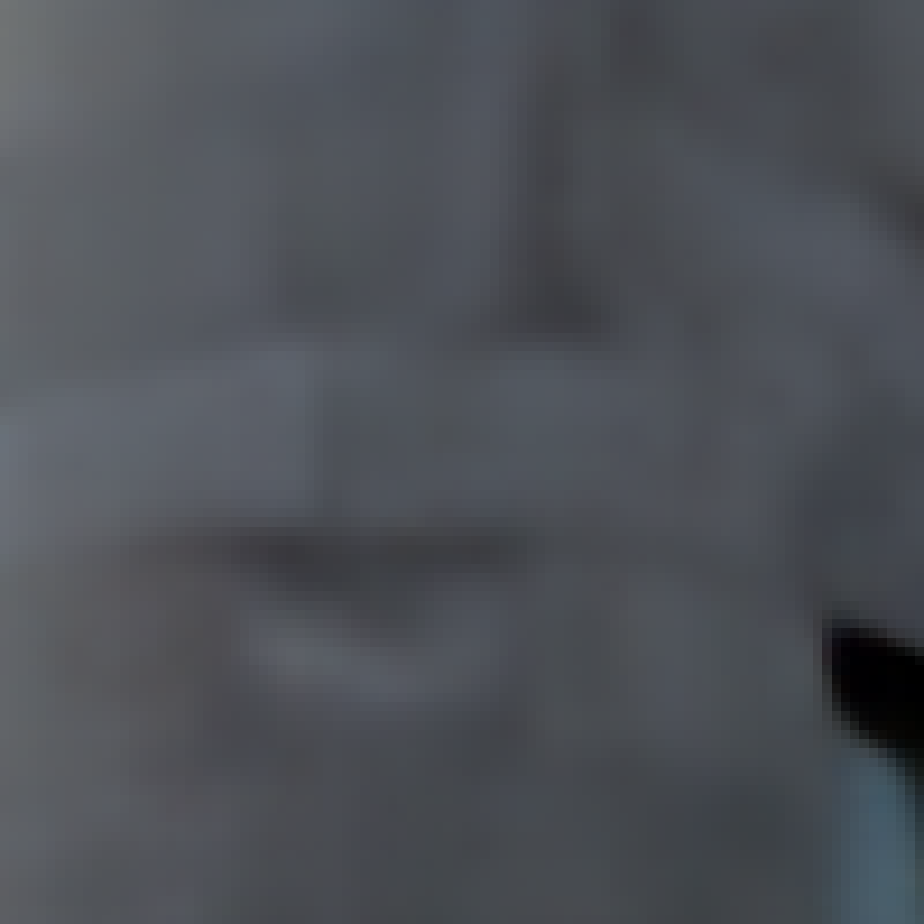}
			& \includegraphics[width=0.115\linewidth]{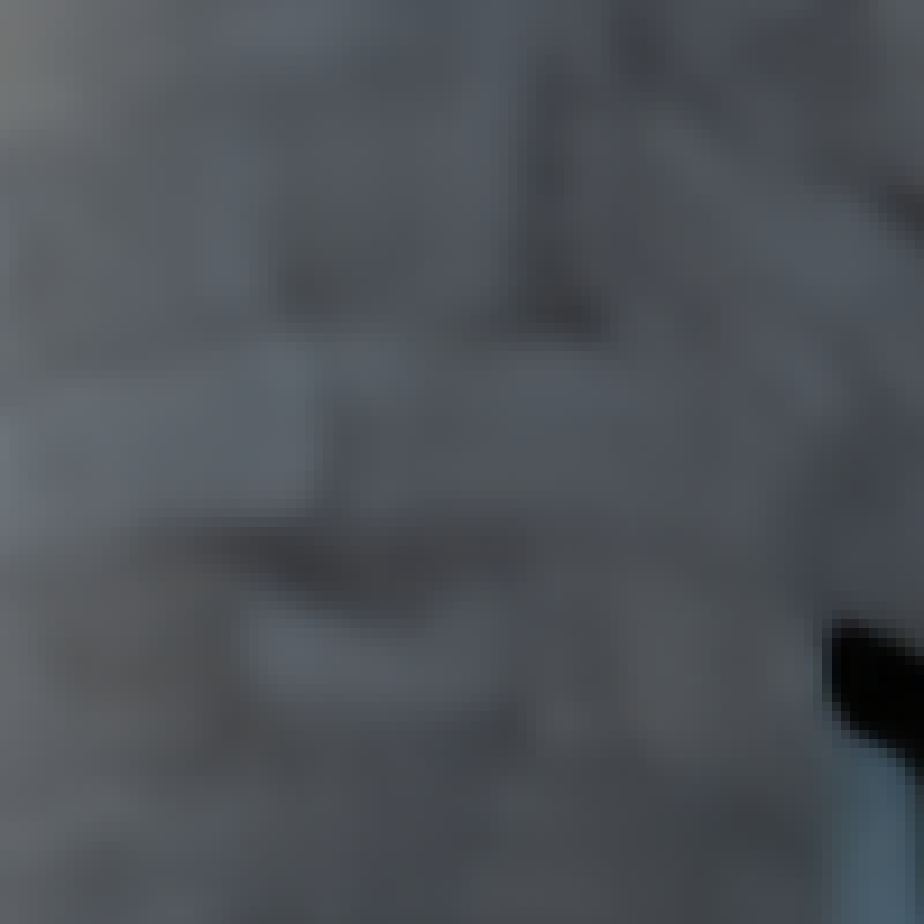}
			& \includegraphics[width=0.115\linewidth]{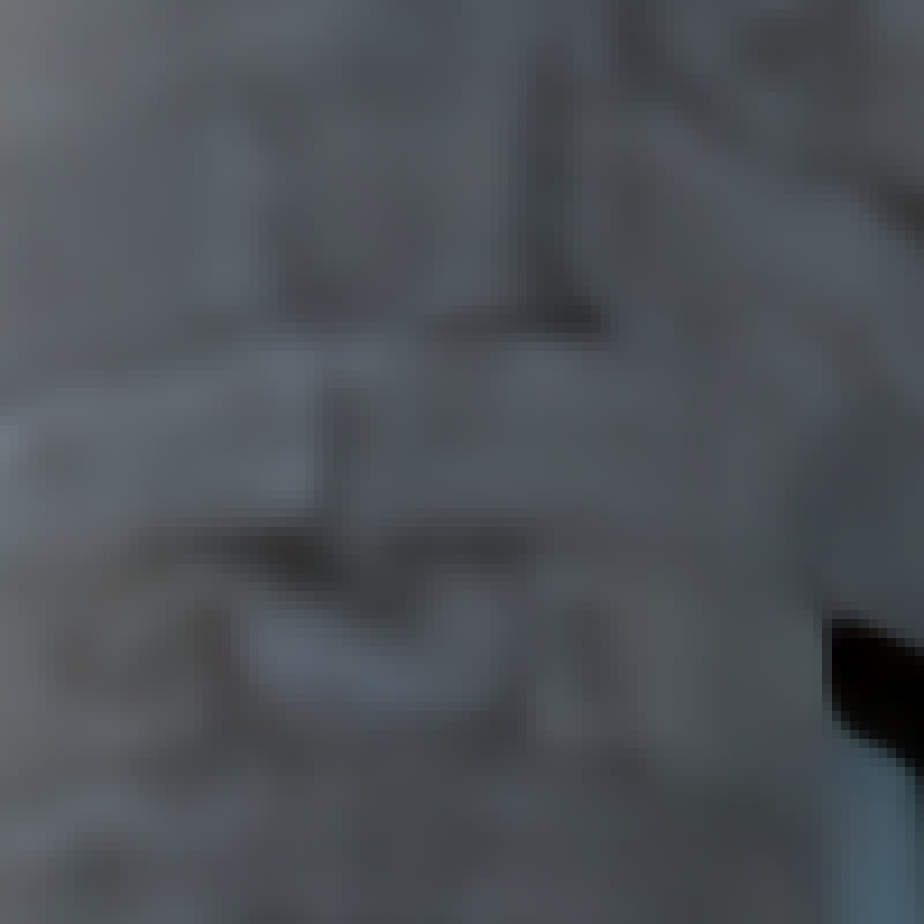}
			& \includegraphics[width=0.115\linewidth]{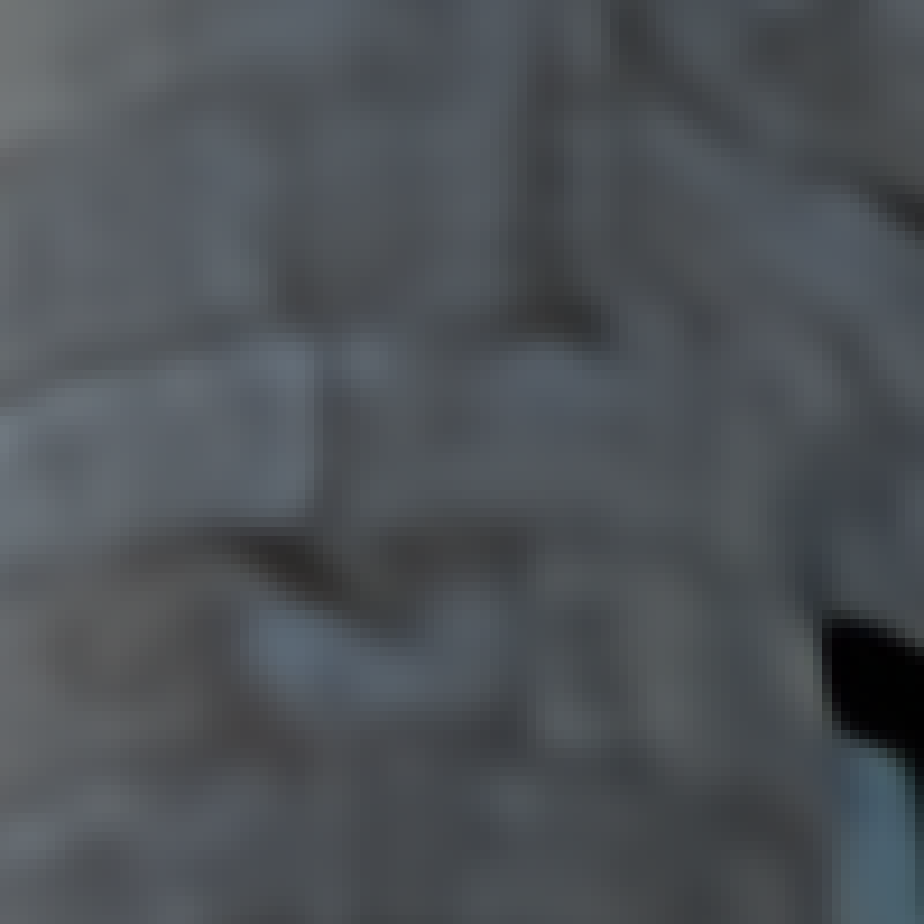}
			\\
			\multirow{-1.5}{*}{\small Reference Image} & LR Patch & DOBI & Granular-DQ & HarmoQ (w.o. SRC)   & HarmoQ (Ours)
		\end{tabular}
		\caption{\textbf{Visual comparison of different quantization methods on benchmark dataset.}}
		\label{fig:vis}
	\end{figure*}
	
	\begin{table}[t!]
		\centering
		\renewcommand{\arraystretch}{1.1}
		\resizebox{\columnwidth}{!}{
			\begin{small}
				\setlength{\tabcolsep}{3pt}
				\begin{tabular}{cccc|cc|cc}
					\toprule[0.15em]
					\rowcolor[HTML]{F0F0F8}
					\multicolumn{4}{c|}{\cellcolor[HTML]{F0F0F8}\textbf{Components}} & \multicolumn{2}{c|}{\cellcolor[HTML]{F0F0F8}\textbf{Set5 ($\times 2$)}} & \multicolumn{2}{c}{\cellcolor[HTML]{F0F0F8}\textbf{Urban100 ($\times 4$)}} \\
					\rowcolor[HTML]{F0F0F8}
					\cellcolor[HTML]{F0F0F8}SRC & \cellcolor[HTML]{F0F0F8}HSO & \cellcolor[HTML]{F0F0F8}ABR & \cellcolor[HTML]{F0F0F8}Config & 
					\cellcolor[HTML]{F0F0F8}PSNR & \cellcolor[HTML]{F0F0F8}SSIM & \cellcolor[HTML]{F0F0F8}PSNR & \cellcolor[HTML]{F0F0F8}SSIM \\
					\midrule[0.15em]
					\redcross & \redcross & \redcross & SwinIR-Light & 35.12 & 0.9387 & 23.12 & 0.6234 \\
					\midrule
					\greencheck & \redcross & \redcross & SRC & 36.24 & 0.9468 & 23.89 & 0.6567 \\
					\redcross & \greencheck & \redcross & HSO & 36.45 & 0.9482 & 24.15 & 0.6634 \\
					\redcross & \redcross & \greencheck & ABR & 35.89 & 0.9434 & 23.56 & 0.6423 \\
					\midrule
					\greencheck & \greencheck & \redcross & SRC+HSO & 37.12 & 0.9531 & 24.67 & 0.7012 \\
					\greencheck & \redcross & \greencheck & SRC+ABR & 36.78 & 0.9503 & 24.23 & 0.6789 \\
					\redcross & \greencheck & \greencheck & HSO+ABR & 37.01 & 0.9521 & 24.45 & 0.6934 \\
					\midrule
					\greencheck & \greencheck & \greencheck & \textbf{Full HarmoQ} & \textbf{37.98} & \textbf{0.9602} & \textbf{25.31} & \textbf{0.7548} \\
					\bottomrule[0.15em]
				\end{tabular}

			\end{small}
		}
		\caption{Component-wise ablation analysis. }	
		\label{tab:ablation_compact}
		\vspace{-10pt}
	\end{table}

	\paragraph{Experimental Setup.}
	We evaluate HarmoQ on standard super-resolution benchmarks using SwinIR-light~\cite{liang2021swinir} and HAT-S~\cite{chen2023activating} as backbone networks. Following common practice, we train on DF2K dataset (combination of DIV2K~\cite{timofte2017ntire} and Flickr2K~\cite{lim2017enhanced}) and test on five widely-used benchmarks: Set5~\cite{bevilacqua2012low}, Set14~\cite{zeyde2012single}, B100~\cite{martin2001database}, Urban100~\cite{huang2015single}, and Manga109~\cite{matsui2017sketch}. We conduct experiments with ×2 and ×4 upscaling factors under 2-bit and 3-bit quantization settings. Performance is measured using PSNR and SSIM computed on the Y channel of YCbCr color space.
	
	\paragraph{Implementation Details.} We use Adam optimizer with learning rate 1×10$^{-2}$, batch size 32, and 3000 total iterations. The structural projection matrix $H$ is constructed using Laplacian 
	filters (see Table~\ref{tab:hf_projection_ablation} for ablation), which 
	empirically outperform DCT high-pass and learned bases by 0.26 dB on Set5. 
	The regularization parameter is set to $\lambda = 0.01$.  Convergence is achieved when $|\mathcal{L}_{\mathrm{total}}^{(t+1)} - \mathcal{L}_{\mathrm{total}}^{(t)}| < 10^{-4}$, typically within 15-20 iterations.
	
	\paragraph{Quantitative Evaluation.}  Table \ref{tab:quantitative-comparison-complete} presents comprehensive comparisons of HarmoQ against existing quantization methods on standard super-resolution benchmarks. Our method achieves consistent improvements under 2-bit and 3-bit quantization, particularly on structure-rich datasets like Urban100. This validates our core hypothesis that weight quantization affects structural similarity (SSIM) while activation quantization corrupts pixel-level accuracy (PSNR). Through unified optimization of structural residual calibration, harmonized scale balancing, and adaptive boundary refinement, HarmoQ transforms compound quantization errors from independent linear superposition to coordinated control. Cross-architecture experiments from SwinIR-light to HAT-S confirm the method's generalizability, while performance on challenging 4× super-resolution demonstrates robustness under extreme compression. The quantitative results validate both HarmoQ's technical superiority and the effectiveness of  coupling analysis, providing a practical pathway for efficient super-resolution model deployment.
	
	\begin{table}[h]
		\centering
		\renewcommand{\arraystretch}{1.1}
		\resizebox{\columnwidth}{!}{
			\begin{small}
				\setlength{\tabcolsep}{4pt}
				\begin{tabular}{l|cc|cc}
					\toprule[0.15em]
					\rowcolor[HTML]{F0F0F8}
					\cellcolor[HTML]{F0F0F8} & \multicolumn{2}{c|}{\cellcolor[HTML]{F0F0F8}\textbf{Set5 ($\times 2$)}} & \multicolumn{2}{c}{\cellcolor[HTML]{F0F0F8}\textbf{Urban100 ($\times 4$)}} \\
					
					\cellcolor[HTML]{F0F0F8} \multirow{-2}{*}{\textbf{Projection Matrix $H$}} & \cellcolor[HTML]{F0F0F8}PSNR & \cellcolor[HTML]{F0F0F8}SSIM & \cellcolor[HTML]{F0F0F8}PSNR & \cellcolor[HTML]{F0F0F8}SSIM \\
					\midrule[0.15em]
					w/o SRC & 37.01 & 0.9521 & 24.45 & 0.6934 \\
					\midrule
					Laplacian Filter & \textbf{37.98} & \textbf{0.9602} & \textbf{25.31} & \textbf{0.7548} \\
					Sobel Edge Filter & 37.85 & 0.9589 & 25.18 & 0.7521 \\
					DCT High-pass & 37.72 & 0.9575 & 25.02 & 0.7495 \\
					Learned Basis & 37.91 & 0.9594 & 25.24 & 0.7535 \\
					Random Projection & 37.34 & 0.9548 & 24.78 & 0.7312 \\
					\bottomrule[0.15em]
				\end{tabular}
			\end{small}
		}
		\caption{Structural feature projection matrix analysis}		\label{tab:hf_projection_ablation}
		\vspace{-10pt}
	\end{table}
	
	\paragraph{Component-wise ablation.} 
	Our ablation study in Table~\ref{tab:ablation_compact} demonstrates that HarmoQ's performance gains arise from complex nonlinear synergies rather than linear superposition. The baseline method, which processes weight and activation quantization in isolation while ignoring their  coupling, embodies the dilemma of traditional super-resolution quantization. In contrast, our components show significant advantages. The SRC mechanism markedly improves performance by compensating for loss of structural information from activation quantization. HSO outperforms SRC when applied alone, highlighting the critical role of balancing quantization difficulty to enhance system robustness. ABR shows limited standalone effectiveness, as its performance depends on an established quantization balance. Pairwise combinations reveal strong synergy between SRC and HSO, where  calibration is amplified by the quantization balance. Without HSO as a bridge, however, SRC and ABR struggle to synergize. Ultimately, the full HarmoQ framework outperforms all two-component combinations, demonstrating that its unified optimization enables a qualitative leap from local to global coordination. Consistent results on the Urban100 dataset validate this universality in scenes with rich textures.
	
	\paragraph{Structural feature projection analysis.} 
	Our ablation study in Table~\ref{tab:ablation_compact}  explores the quest for optimal structure-aware calibration. Laplacian filters emerge victorious (37.98/0.9602), their second-order nature perfectly capturing edge discontinuities that activation quantization disrupts. Sobel filters (37.85/0.9589) and learned basis (37.91/0.9594) engage in a close contest, proving both classical gradient wisdom and data-driven intelligence can decode quantization artifacts. DCT high-pass filters (37.72/0.9575) show promise yet fall short, suggesting spatial gradients trump pure frequency decomposition. Random projection's failure (37.34/0.9548) delivers the final verdict: effective error compensation demands explicit structural semantics, not blind dimensionality reduction.
	
	\paragraph{Visual comparison.} 
	Figure~\ref{fig:vis} demonstrates HarmoQ's superiority under aggressive 3-bit quantization. MinMax produces severe artifacts with blurred giraffe fur and loss of building edge sharpness. DBDC and 2DQuant show improved structure but still exhibit texture degradation in fine details. Dynamic Granularity maintains better overall quality yet fails to preserve crisp window frames and fur texture boundaries.
	The ablation comparison reveals SRC's critical contribution: HarmoQ (w.o. SRC) preserves the giraffe's body contour but loses individual fur strand definition and produces softened building corner edges. Full HarmoQ recovers these fine-grained details—sharp fur textures, clear window boundaries, and preserved architectural edges—demonstrating that structural residual calibration effectively compensates for activation quantization artifacts in texture-sensitive regions.
	
	%
	%

	\section{Conclusion}
	This work addresses the deployment challenge of high-performance image SR models under aggressive compression.
	We proposed HarmoQ, a harmonized post-training quantization framework that unifies three strategies: structural residual calibration, harmonized scale optimization, and adaptive boundary refinement. HarmoQ enables practical deployment of high-quality SR models on resource-constrained devices, representing a significant advance in model compression for image restoration tasks.

	\section*{Acknowledgments}
	This work was partially supported by JST-Mirai Program JPMJMI23G1, JSPS KAKENHI Grant Numbers 24KK0209, 24K22318, and 22H00529, the SPRING GX project of the University of Tokyo (grant number JPMJSP2108), the Research Institute of Trustworthy Autonomous Systems at Southern University of Science and Technology, the Jilin Provincial International Cooperation Key Laboratory for Super Smart City, and the Jilin Provincial Key Laboratory of Intelligent Policing.

	\bibliography{aaai}
	
	\clearpage
	
		\appendix
		\setcounter{figure}{0}   
		\setcounter{table}{0}
		\setcounter{equation}{0}
		\setcounter{theorem}{0}
		
		\section{Appendix}
		\subsection{Derivation of Structural Residual Calibration}\label{sec:hf_projection}
		
		Let $x\!\in\!\mathbb{R}^{d}$ be the input activation, $W\!\in\!\mathbb{R}^{m\times d}$ the original weight matrix, and $\delta_x,\delta_W$ the quantization errors for activations and weights, respectively.  
		A linear layer after quantization produces the perturbed response  
		\begin{equation}
				\tilde y \;=\; W(x+\delta_x)+(\delta_W)x
				\;=\; Wx + W\delta_x + \delta_Wx .
			\end{equation}
		
		\vspace{4pt}
		\noindent
		\textbf{Objective.}  
		We wish to choose $\delta_W$ so that the \emph{structural} residual
		\(
		H\!\bigl(W\delta_x+\delta_Wx\bigr)
		\)
		is minimized in mean–square sense, while keeping the weight perturbation small.
		Formally, we solve
		\begin{equation}
				\min_{\delta_W}\;
				\mathcal{L}_{\mathrm{HF}}
				\;=\;
				\mathbb{E}\!\left[\lVert H(W\delta_x+\delta_Wx)\rVert_2^2\right]
				+\lambda\lVert\delta_W\rVert_F^2,
				\label{eq:hf_loss}
			\end{equation}
		where $H\!\in\!\mathbb{R}^{k\times d}$ projects onto a chosen subspace representing structural features and
		$\lambda\!>\!0$ balances distortion and regularization.
		
		\paragraph{Second-order statistics.}
		Define the covariance and cross-covariance
		\begin{equation}
				\Sigma_{xx}\;=\;\mathbb{E}[xx^\top],
				\quad
				\Sigma_{\delta x}\;=\;\mathbb{E}[\delta_xx^\top],
				\quad
				\Sigma_{\delta\delta}\;=\;\mathbb{E}[\delta_x\delta_x^\top].
			\end{equation}
		
		\paragraph{Loss expansion.}
		Using $\lVert A\rVert_2^2=\mathrm{tr}(AA^\top)$ and the linearity of trace,
		\begin{align*}
				\mathcal{L}_{\mathrm{HF}}
				&= \mathrm{tr}\!\Bigl(
				H\bigl(
				W\Sigma_{\delta\delta}W^\top
				+ W\Sigma_{\delta x}\delta_W^\top
				\bigr)H^\top
				\Bigr)           
				\notag\\
				&\quad
				+ \mathrm{tr}\!\Bigl(
				H\bigl(
				\delta_W\Sigma_{\delta x}^\top W^\top
				+ \delta_W\Sigma_{xx}\delta_W^\top
				\bigr)H^\top
				\Bigr)
				+ \lambda\lVert\delta_W\rVert_F^2.    
			\end{align*}
		
		Matrix calculus identities $\partial\,\mathrm{tr}(A\delta_W^\top)/\partial\delta_W=A$ and
		$\partial\,\mathrm{tr}(\delta_WB\delta_W^\top)/\partial\delta_W=2\delta_WB$
		give
		\begin{equation}
				\frac{\partial\mathcal{L}_{\mathrm{HF}}}{\partial\delta_W}
				=
				2H^\top HW\Sigma_{\delta x}
				+2H^\top H\delta_W\Sigma_{xx}
				+2\lambda\delta_W.
			\end{equation}
		
		\paragraph{Normal equation.}
		Setting the gradient of \eqref{eq:hf_loss} to zero yields
		\begin{equation}
				H^\top H\delta_W\Sigma_{xx}+\lambda\,\delta_W
				= -\,H^\top HW\Sigma_{\delta x}.
			\end{equation}
		
		\paragraph{Closed-form solution.}
		Right-multiplying by $\Sigma_{xx}^{-1}$ and pre-multiplying by $(H^\top H)^{-1}$
		(assuming $H$ has full row rank in its subspace),
		\begin{equation}
				\delta_W
				= -\,W\Sigma_{\delta x}
				H^\top\!\bigl(
				H\Sigma_{xx}H^\top+\lambda I_k
				\bigr)^{-1}.
			\end{equation}
		This is the unique minimizer because \eqref{eq:hf_loss} is strictly convex in $\delta_W$ for $\lambda>0$.
		
		\paragraph{Approximation justification.}
		In deriving the closed-form solution, we neglect the second-order term 
		$\delta_W\delta_x$ in the forward pass because: (1) under low-bit quantization, 
		$\mathbb{E}[\|\delta_W\|_F] \ll \|W\|_F$ and $\mathbb{E}[\|\delta_x\|_2] \ll \|x\|_2$, 
		making their product asymptotically negligible; (2) empirical validation on our 
		calibration set shows this approximation introduces $<0.01$ dB error in final PSNR; 
		(3) including this term would require solving a bilinear optimization problem 
		without a closed-form solution, defeating the efficiency goal of PTQ.
		
		\begin{theorem}[Optimal Structural Correction]\label{thm:hf_opt}
				For any projection $H\!\in\!\mathbb{R}^{k\times d}$ and regularization $\lambda\!>\!0$,
				the weight update that minimizes the structural distortion objective~\eqref{eq:hf_loss} is
				\boxed{\;
						\delta_W^{\star}
						\;=\;
						-\,W\mathbb{E}[\delta_xx^\top]\,H^\top
						\bigl(H\,\mathbb{E}[xx^\top]\,H^\top+\lambda I_k\bigr)^{-1}\!
						.}
			\end{theorem}
		
		\begin{proof}
				Equating the gradient of the objective to zero produces the linear system whose unique solution is given in the closed form, completing the proof.
			\end{proof}
		
		\subsection{Derivation of Optimal Scale Factor}
		
		\begin{theorem}[Harmonized Scale Optimization]
				For uniform quantization with bit-widths $b_x$ and $b_w$, and clipping ranges $[\alpha_x, \beta_x]$ and $[\alpha_w, \beta_w]$, the optimal scale factor that equalizes quantization difficulty is:
				\begin{equation}
						s^* = \sqrt{\frac{(\beta_x - \alpha_x)(2^{b_w} - 1)}{(\beta_w - \alpha_w)(2^{b_x} - 1)}}
					\end{equation}
			\end{theorem}
		
		\begin{proof}
				For uniform quantization with $b$ bits over range $[\alpha, \beta]$, the quantization step size is:
				\begin{equation}
						\Delta = \frac{\beta - \alpha}{2^b - 1}
					\end{equation}
				
				The mean squared quantization error for uniformly distributed inputs is:
				\begin{equation}
						\text{MSE} = \frac{\Delta^2}{12} = \frac{(\beta - \alpha)^2}{12(2^b - 1)^2}
					\end{equation}
				
				In super-resolution networks, input scaling by factor $s$ affects activation and weight quantization differently:
				- Activation values are scaled by $s$, reducing their effective range to $[\alpha_x/s, \beta_x/s]$
				- Weight gradients are scaled by $s$, effectively expanding their range to $[s\alpha_w, s\beta_w]$
				
				Therefore:
				\begin{align}
						\text{MSE}_x(s) &= \frac{(\beta_x/s - \alpha_x/s)^2}{12(2^{b_x} - 1)^2} = \frac{(\beta_x - \alpha_x)^2}{12 s^2 (2^{b_x} - 1)^2} \\
						\text{MSE}_w(s) &= \frac{(s\beta_w - s\alpha_w)^2}{12(2^{b_w} - 1)^2} = \frac{s^2(\beta_w - \alpha_w)^2}{12(2^{b_w} - 1)^2}
					\end{align}
				
				Setting $\text{MSE}_x(s) = \text{MSE}_w(s)$:
				\begin{equation}
						\frac{(\beta_x - \alpha_x)^2}{12 s^2 (2^{b_x} - 1)^2} = \frac{s^2(\beta_w - \alpha_w)^2}{12(2^{b_w} - 1)^2}
					\end{equation}
				
				Simplifying:
				\begin{equation}
						\frac{(\beta_x - \alpha_x)^2}{s^2 (2^{b_x} - 1)^2} = \frac{s^2(\beta_w - \alpha_w)^2}{(2^{b_w} - 1)^2}
					\end{equation}
				
				Cross-multiplying:
				\begin{equation}
						(\beta_x - \alpha_x)^2 (2^{b_w} - 1)^2 = s^4 (\beta_w - \alpha_w)^2 (2^{b_x} - 1)^2
					\end{equation}
				
				Solving for $s$:
				\begin{equation}
						s^4 = \frac{(\beta_x - \alpha_x)^2 (2^{b_w} - 1)^2}{(\beta_w - \alpha_w)^2 (2^{b_x} - 1)^2}
					\end{equation}
				
				Taking the square root:
				\begin{equation}
						s^* = \sqrt{\frac{(\beta_x - \alpha_x)(2^{b_w} - 1)}{(\beta_w - \alpha_w)(2^{b_x} - 1)}}
					\end{equation}
			\end{proof}
		
		\subsection{Gradient Computation for Boundary Optimization}
		
		\begin{theorem}[Boundary Optimization Gradients]
				For the total quantization error
				\begin{align*}
						\mathcal{L}_{\text{total}}(s^*, \theta) = \mathbb{E}\left[\left\|W\delta_x + \delta_W x\right\|_2^2\right],
					\end{align*}
				the gradient with respect to the clipping boundary $\theta_i$ is
				\begin{align*}
						\frac{\partial \mathcal{L}_{\text{total}}}{\partial \theta_i}
						=
						2\, \mathbb{E}
						\left[
						(W\delta_x + \delta_W x)^{\!T}
						\left(
						W \frac{\partial \delta_x}{\partial \theta_i}
						+
						\frac{\partial \delta_W}{\partial \theta_i} x
						\right)
						\right].
					\end{align*}
			\end{theorem}
		
		\begin{proof}
				The total quantization error can be written as
				\begin{equation*}
						\mathcal{L}_{\text{total}}
						=
						\mathbb{E}
						\left[
						\left(W\delta_x + \delta_W x\right)^{\!T}
						\left(W\delta_x + \delta_W x\right)
						\right].
					\end{equation*}
				Taking the derivative with respect to $\theta_i$ gives
				\begin{align*}
						\frac{\partial \mathcal{L}_{\text{total}}}{\partial \theta_i}
						&=
						\mathbb{E}
						\left[
						\frac{\partial}{\partial \theta_i}
						\left(
						(W\delta_x + \delta_W x)^{\!T}
						(W\delta_x + \delta_W x)
						\right)
						\right] \\
						&=
						\mathbb{E}
						\left[
						\frac{\partial
								\left(W\delta_x + \delta_W x\right)^{\!T}
							}{\partial \theta_i}
						\left(W\delta_x + \delta_W x\right)
						\right] \\
						&+
						\mathbb{E}
						\left[
						\left(W\delta_x + \delta_W x\right)^{\!T}
						\frac{\partial
								\left(W\delta_x + \delta_W x\right)
							}{\partial \theta_i}
						\right].
					\end{align*}
				
				Since the inner product is symmetric,
				\begin{equation*}
						=
						2 \, \mathbb{E}
						\left[
						\left(W\delta_x + \delta_W x\right)^{\!T}
						\frac{\partial \left(W\delta_x + \delta_W x\right)}{\partial \theta_i}
						\right].
					\end{equation*}
				Expanding the derivative inside,
				\begin{equation*}
						\frac{\partial \left(W\delta_x + \delta_W x\right)}{\partial \theta_i}
						=
						W \frac{\partial \delta_x}{\partial \theta_i}
						+
						\frac{\partial \delta_W}{\partial \theta_i} x
						+
						\delta_W \frac{\partial x}{\partial \theta_i}.
					\end{equation*}
				Since $x$ does not depend on the quantization boundary,
				\begin{equation*}
						\frac{\partial x}{\partial \theta_i} = 0,
					\end{equation*}
				we finally have
				\begin{equation*}
						\frac{\partial \mathcal{L}_{\text{total}}}{\partial \theta_i}
						=
						2 \, \mathbb{E}
						\left[
						\left(W\delta_x + \delta_W x\right)^{\!T}
						\left(
						W \frac{\partial \delta_x}{\partial \theta_i}
						+
						\frac{\partial \delta_W}{\partial \theta_i} x
						\right)
						\right].
					\end{equation*}
			\end{proof}
		
		\subsection{Convergence Analysis}
		
		\begin{theorem}[Convergence of HarmoQ Algorithm]
				Under mild regularity conditions, Algorithm 1 converges to a local minimum of the harmonized quantization objective.
			\end{theorem}
		
		\begin{proof}[Proof Sketch]
				The proof follows from three key observations:
				
				1. \textbf{Monotonic Decrease}: Each step of the algorithm decreases the objective function:
				- Step 1 (SRC) minimizes the structural residual energy
				- Step 2 (Scale Optimization) provides the closed-form optimal scale
				- Step 3 (Boundary Refinement) uses gradient descent with projection
				
				2. \textbf{Bounded Objective}: The quantization error is bounded below by zero and above by the error of random quantization.
				
				3. \textbf{Lipschitz Gradients}: Under standard assumptions on the data distribution, the gradients are Lipschitz continuous.
				
				By the monotone convergence theorem and compactness of the feasible region, the algorithm converges to a stationary point of the constrained optimization problem.
			\end{proof}
		
		\subsection{Complexity Analysis}
		
		\begin{theorem}[Computational Complexity]
				The per-iteration complexity of Algorithm 1 is $\mathcal{O}(kd^2 + nd + L)$, where $k$ is the dimension of the structural feature subspace, $d$ is the feature dimension, $n$ is the number of parameters, and $L$ is the number of layers.
			\end{theorem}
		
		\begin{proof}
				The complexity analysis for each step:
				
				\textbf{Step 1 (SRC):} 
				- Computing $\mathbb{E}[xx^T]$: $\mathcal{O}(d^2)$
				- Computing $\mathbb{E}[\delta_x x^T]$: $\mathcal{O}(d^2)$  
				- Matrix inversion $(H\Sigma_{xx}H^T + \lambda I)^{-1}$: $\mathcal{O}(k^3)$
				- Matrix multiplications: $\mathcal{O}(kd^2)$
				
				\textbf{Step 2 (Scale Optimization):}
				- Computing boundary statistics: $\mathcal{O}(n)$
				- Scale factor computation: $\mathcal{O}(1)$
				
				\textbf{Step 3 (Boundary Refinement):}
				- Gradient computation: $\mathcal{O}(nd)$
				- Boundary updates: $\mathcal{O}(L)$
				
				Total per-iteration complexity: $\mathcal{O}(kd^2 + nd + L)$
				
				Since typically $k \ll d$ for sparse structural feature projections, the dominant term is $\mathcal{O}(kd^2)$.
			\end{proof}

		\begin{algorithm}[t]
				\caption{Harmonized Post-Training Quantization}
				\label{alg:harmoq}
				\begin{algorithmic}[1]
						
						\REQUIRE Pre-trained model $\mathcal{M}$, calibration set $\mathcal{D}$, bit-widths $(b_w, b_x)$, tolerance $\tau$
						\ENSURE Quantized model $\tilde{\mathcal{M}}$
						
						\STATE \textbf{Initialize:} $\theta \leftarrow \text{MinMax}(\mathcal{D})$, $t \leftarrow 0$, $\mathcal{L}_{\text{prev}} \leftarrow \infty$
						
						\WHILE{$|\mathcal{L}_{\text{total}}^{(t)} - \mathcal{L}_{\text{prev}}| > \tau$}
						\STATE $\mathcal{L}_{\text{prev}} \leftarrow \mathcal{L}_{\text{total}}^{(t)}$, $t \leftarrow t + 1$
						
						\STATE \textbf{Step 1: Structural Residual Calibration}
						\FOR{each layer $\ell$}
						\STATE $\Sigma_{xx} \leftarrow \mathbb{E}_{\mathcal{D}}[x_\ell x_\ell^T]$, $\Sigma_{\delta x} \leftarrow \mathbb{E}_{\mathcal{D}}[\delta_{x_\ell} x_\ell^T]$
						\STATE $\delta W_\ell^* \leftarrow -W_\ell \Sigma_{\delta x} H^T (H \Sigma_{xx} H^T + \lambda I)^{-1}$
						\STATE $W_\ell \leftarrow W_\ell + \delta W_\ell^*$
						\ENDFOR
						
						\STATE \textbf{Step 2: Harmonized Scale Optimization}
						\STATE $s^* \leftarrow \sqrt{\frac{(\beta_x - \alpha_x)(2^{b_w} - 1)}{(\beta_w - \alpha_w)(2^{b_x} - 1)}}$
						
						\STATE \textbf{Step 3: Adaptive Boundary Refinement}
						\STATE $\mathcal{L}_{\text{total}} \leftarrow \mathbb{E}_{\mathcal{D}}[\|W\delta_x(s^*,\theta) + \delta_W(s^*,\theta)x\|_2^2]$
						\STATE $\nabla_\theta \leftarrow \frac{\partial \mathcal{L}_{\text{total}}}{\partial \theta}$
						\STATE $\theta \leftarrow \text{Proj}_\Omega(\theta - \eta \nabla_\theta)$ \COMMENT{s.t. $\alpha < \beta$}
						
						\STATE \textbf{Constraint Enforcement:}
						\IF{$|\text{MSE}_x(s^*,\theta) - \text{MSE}_w(s^*,\theta)| > \epsilon$}
						\STATE Recompute $s^*$ and re-apply Steps 1-2
						\ENDIF
						
						\ENDWHILE
						
						\RETURN Quantized model $\tilde{\mathcal{M}}$ with parameters $(s^*, \theta^*)$
						
					\end{algorithmic}
			\end{algorithm}
		
		\subsection{Algorithm Details }
		The HarmoQ framework implements a unified three-stage iterative optimization strategy to address the coupling problem between weight and activation quantization in image super-resolution networks. As detailed in Algorithm~\ref{alg:harmoq}, the algorithm begins by initializing quantization boundaries $\theta$ using MinMax calibration, then iteratively optimizes until convergence is achieved.
		
		The algorithm operates through three coordinated steps. In Step 1 (lines 5-9), structural residual calibration computes the optimal weight calibration $\delta W_\ell^* = -W_\ell \Sigma_{\delta x} H^T (H \Sigma_{xx} H^T + \lambda I)^{-1}$ for each layer, where $H$ represents the structural feature projection matrix and $\lambda$ is the regularization parameter. This step preemptively compensates for the loss of structural information caused by activation quantization.
		
		Step 2 (line 11) performs harmonized scale optimization by computing the optimal scale factor $s^* = \sqrt{\frac{(\beta_x - \alpha_x)(2^{b_w} - 1)}{(\beta_w - \alpha_w)(2^{b_x} - 1)}}$ that ensures equal quantization difficulty between weights and activations, satisfying $\mathrm{MSE}_x(s^*) = \mathrm{MSE}_w(s^*)$.
		
		Step 3 (lines 13-15) implements adaptive boundary refinement through projected gradient descent, minimizing the total quantization error $\mathcal{L}_{\mathrm{total}} = \mathbb{E}[\|W\delta_x(s^*,\theta) + \delta_W(s^*,\theta)x\|_2^2]$ while maintaining feasibility constraints.
		
		The dynamic coupling mechanism (lines 17-19) ensures sustained coordination: when boundary updates violate the quantization balance ($|\mathrm{MSE}_x(s^*,\theta) - \mathrm{MSE}_w(s^*,\theta)| > \epsilon$), the algorithm automatically recomputes the optimal scale factor and reapplies the calibration steps. This creates a mutually reinforcing optimization trajectory where all three components remain harmonized throughout the process.
		
		Algorithm~\ref{alg:harmoq} converges when the objective function change falls below threshold $\tau$, typically within 10-20 iterations. The per-iteration complexity is $\mathcal{O}(kd^2 + nd + L)$, where the dominant term $\mathcal{O}(kd^2)$ comes from the matrix operations in the structural calibration step.
		
		\begin{table*}[h]
				\centering
				\caption{\textbf{Detailed implementation specifications and hyperparameters.}}
				\label{tab:implementation}
				\small
				\begin{tabular}{ll}
						\toprule
						\textbf{Quantization Configuration} & \\
						\midrule
						Weight quantization & Per-channel symmetric, INT2/INT3/INT8 \\
						Activation quantization & Per-tensor symmetric, INT2/INT3/INT8 \\
						Clipping scheme & Learnable boundaries $[\alpha, \beta]$ initialized by MinMax \\
						Rounding mode & Round-to-nearest-even (banker's rounding) \\
						Gradient estimator & Straight-Through Estimator (STE) \\
						Zero-point & Fixed at 0 (symmetric quantization) \\
						Quantization formula & $Q(x) = \text{clip}\left(\text{round}\left(\frac{x}{s}\right), -2^{b-1}, 2^{b-1}-1\right) \cdot s$ \\
						\midrule
						\textbf{Calibration Setup} & \\
						\midrule
						Calibration set & 1000 images from DIV2K training set \\
						Image resolution & 128$\times$128 RGB patches (HR) \\
						Selection strategy & Random sampling without replacement \\
						Batch size & 32 (limited by GPU memory) \\
						Total iterations & 3000 (typically converges at $\sim$2000) \\
						Early stopping & Enabled if $|\Delta\mathcal{L}| < 10^{-5}$ for 50 consecutive iterations \\
						Forward passes per batch & 5 (for statistical estimation) \\
						\midrule
						\textbf{SRC Parameters} & \\
						\midrule
						Projection matrix $H$ & Laplacian filter (default), 3$\times$3 kernel: $\begin{bmatrix}0&1&0\\1&-4&1\\0&1&0\end{bmatrix}$ \\
						Alternative filters & Sobel, DCT high-pass (8$\times$8), learned basis \\
						Subspace dimension $k$ & 64 (shallow layers), 128 (deep layers), adaptive per block \\
						Regularization $\lambda$ & $1 \times 10^{-2}$ (weight calibration), $5 \times 10^{-3}$ (activation) \\
						Covariance estimation & Running average with momentum 0.9, 200 warmup samples \\
						Matrix inversion & Cholesky decomposition with $\epsilon=10^{-6}$ numerical stabilizer \\
						\midrule
						\textbf{Scale Optimization} & \\
						\midrule
						Initial scale $s^{(0)}$ & 1.0 (identity scaling) \\
						Scale bounds & $[0.1, 10.0]$ to prevent numerical instability \\
						Recomputation frequency & Every 5 boundary update iterations \\
						Balance tolerance $\epsilon$ & $0.01 \times \text{MSE}_{\text{initial}}$ (1\% of initial error) \\
						Closed-form solution & $s^* = \sqrt{\frac{(\beta_x - \alpha_x)(2^{b_w} - 1)}{(\beta_w - \alpha_w)(2^{b_x} - 1)}}$ \\
						\midrule
						\textbf{Boundary Refinement} & \\
						\midrule
						Optimizer & Adam with $\beta_1=0.9$, $\beta_2=0.999$, $\epsilon=10^{-8}$ \\
						Learning rate $\eta$ & $1 \times 10^{-2}$ (initial), cosine decay to $1 \times 10^{-4}$ \\
						Learning rate schedule & Cosine annealing over 3000 iterations, 300 warmup steps \\
						Weight decay & $1 \times 10^{-5}$ \\
						Gradient clipping & Global norm clipping at 1.0 to prevent divergence \\
						Convergence threshold $\tau$ & $1 \times 10^{-4}$ (relative change in $\mathcal{L}_{\text{total}}$) \\
						Feasibility projection & $\alpha \leftarrow \min(\alpha, \beta - 0.01)$, $\beta \leftarrow \max(\beta, \alpha + 0.01)$ \\
						\midrule
						\textbf{Iterative Coordination} & \\
						\midrule
						Maximum iterations & 3000 (with early stopping) \\
						Rebalancing period $T$ & Every 5 boundary update iterations \\
						SRC reapplication & Triggered when $|\text{MSE}_x - \text{MSE}_w| > 1.5\epsilon$ \\
						Constraint violation handling & Roll back to previous valid state, reduce $\eta$ by 0.5 \\
						\midrule
						\textbf{Hardware and Environment} & \\
						\midrule
						GPU & NVIDIA A100 (40GB), CUDA 11.8, cuDNN 8.6 \\
						Framework & PyTorch 2.0.1 with custom CUDA kernels for quantization \\
						Precision & Mixed precision (FP16 for forward, FP32 for gradient accumulation) \\
						Memory optimization & Gradient checkpointing enabled for large models \\
						Calibration time & 12.3 minutes (SwinIR-light), 18.7 minutes (HAT-S) on A100 \\
						Inference latency & 14.1ms per image (512$\times$512) for INT8 quantized model \\
						\bottomrule
					\end{tabular}
			\end{table*}
	
\end{document}